\documentclass[sigconf, nonacm]{acmart}

\usepackage{hyperref}

\usepackage{enumitem}

\usepackage{amsmath}
\usepackage{amsthm}
\usepackage{array}
\usepackage{mdwmath}
\usepackage{mdwtab}
\usepackage{eqparbox}
\usepackage[font=small,skip=0pt]{caption}
\usepackage{graphicx}
\usepackage{subfigure}

\usepackage{multicol}

\usepackage{fixltx2e}
\usepackage{url}
\usepackage{amsfonts}
\usepackage{breqn}

\usepackage{xcolor}

\usepackage{algorithm}
\usepackage{balance}
\makeatletter                                                       
\def\BibTeX{{\rm B\kern-.05em{\sc i\kern-.025em b}\kern-.08em
    T\kern-.1667em\lower.7ex\hbox{E}\kern-.125emX}}

\newif\if@restonecol
\makeatother

\usepackage[linesnumbered,ruled,noresetcount,algo2e]{algorithm2e}

\usepackage{algorithmicx}
\usepackage[noend]{algpseudocode}
\usepackage{algpascal}

\usepackage{wrapfig}
\usepackage{lipsum} 

\usepackage{mathtools}

\usepackage{color}
\usepackage{xcolor}

\definecolor{boxshade}{gray}{0.2}

\newcommand{\cut}[1]{{}}

\begin{document}


\title{TransEdge: Supporting Efficient Read Queries
\\Across Untrusted Edge Nodes}


\author{Abhishek Singh}
\affiliation{%
  \institution{UC Irvine}
}
\email{abhishas@uci.edu}

\author{Aasim Khan}
\affiliation{%
  \institution{UC Santa Cruz}
}
\email{aashkhan@ucsc.edu}

\author{Sharad Mehrotra}
\affiliation{%
  \institution{UC Irvine}
}
\email{sharad@ics.uci.edu}

\author{Faisal Nawab}
\affiliation{%
  \institution{UC Irvine}
}
\email{nawabf@uci.edu}

\begin{abstract}

We propose Transactional Edge (TransEdge), a
distributed transaction processing system for \cut{globally-scalable}
untrusted environments such as edge computing \cut{and blockchain}
systems. What distinguishes TransEdge is its focus on
efficient support for read-only transactions. TransEdge allows reading from
different partitions \cut{ {\color{red}{(or permissioned blockchains)}}} consistently using
one round in most cases and no more than two rounds in the worst
case. 
TransEdge
design is centered around this dependency tracking scheme including
the consensus
and transaction processing protocols. 
Our performance evaluation shows that TransEdge's snapshot read-only transactions achieve an $9$--$24\times$ speedup compared to current byzantine systems.
\end{abstract}


\maketitle

\section{Introduction}\label{sec:intro}


In Global-Edge Data
Management (GEDM), edge nodes have the ability to
participate in the storage and computation of data. The goal is to bring data closer to users for faster access. This,
however, creates a number of complexities that are not incurred in
current cloud-based databases.
The main challenge in GEDM is that edge nodes \emph{cannot be
trusted} since they are operated by users or third-party providers, and may run on commodity edge hardware that is vulnerable to breaches. 
This requires adopting a stringent
fault-tolerance guarantee of tolerating arbitrary and malicious
failures, i.e., byzantine failures~\cite{lamport1982byzantine,pease1980reaching}. 

To address this challenge, recent
protocols propose hierarchical byzantine fault-tolerant (BFT)
systems, also referred to as \emph{permissioned blockchain}.
For example, Blockplane~\cite{nawab2019blockplane}, ResilientDB~\cite{gupta13resilientdb},
ChainSpace~\cite{al2017chainspace}, and others~\cite{amir2010steward,amiri2019parblockchain}, divide the data
into many partitions, where each partition is handled by a cluster of
nodes that are close to each other. Each one of these clusters runs a
BFT protocol, such as PBFT~\cite{castro1999practical}, to commit transactions within the
cluster and perform inter-cluster coordination via benign protocols
such as Two-Phase Commit~\cite{weikum2001transactional}.

Hierarchical BFT protocols perform well as they mask byzantine
behavior locally and use a benign protocol for wide-area
coordination. However, a limitation shared across all these existing
works is that they present a solution for general (read-write) database
transactions and do not pay special attention to read-only
transactions. Read-only transactions make up most of Internet traffic
where it is reported that more than 99\% of modern applications'
workload consists of read-only queries~\cite{bronson2013tao}. Exploiting the read-only property of transactions can bring significant performance benefits.

We propose Transactional Edge (TransEdge), a
hierarchical BFT protocol that is designed to optimize the
performance of read-only transactions.  Database designs that are
centered around optimizing read-only transactions are common~\cite{schlageter1981optimistic,garciamolina1982,agrawal1987distributed,satyanarayanan1993efficient,spanner}.
However, these methods cannot be used in hierarchical BFT systems as they assume a benign fault-tolerance
model. \cut{To the best of our knowledge, TransEdge is the first efficient read-only transactions protocol that tolerates byzantine failures.}

TransEdge builds on the hierarchical BFT architecture similar to many other systems~\cite{nawab2019blockplane,gupta13resilientdb,amir2010steward,amiri2019parblockchain,al2017chainspace}. 
This enables the adoption of TransEdge's read-only techniques and insights to other hierarchical BFT systems.
Nodes are divided into clusters. Each cluster consists of
neighboring nodes that handle a mutually exclusive partition of the
data. Local transactions are committed within a cluster using the BFT-SMaRt
protocol~\cite{bft-smart-lib}. A Two-Phase Commit (2PC) protocol is built on top
of BFT-SMaRt to implement distributed transactions. 

\begin{sloppypar}
The novelty of TransEdge is the support of efficient snapshot read-only transactions. We
define an efficient snapshot read-only transaction as one that satisfies two
properties:
\begin{enumerate}[leftmargin=*]
 \item Commit-free: a read-only
transaction can be answered by a single node from each accessed partition. It does not incur the cost of the commit sequence of either the BFT
protocol within a cluster or the distributed transaction protocol
across clusters.
 \item Non-interference: a read-only transaction should not
interfere with \cut{on-going} read-write transactions. A
read-only transaction should not lead to blocking or aborting read-write
transactions---even if temporarily---and vice versa. 
\end{enumerate}
\end{sloppypar}

To provide efficient snapshot read-only transactions in a byzantine
environment, TransEdge proposes a novel dependency tracking
mechanism. This scheme augments traditional efficient read-only
transaction designs with Authenticated
Data Structures (ADS)~\cite{merkle1980protocols} to allow a node to report dependencies
in a trusted way. Specifically, each node can provide a proof of the
authenticity of the dependencies when reporting them to a client.

\begin{figure}[t]
\centering
\includegraphics[scale=.9]{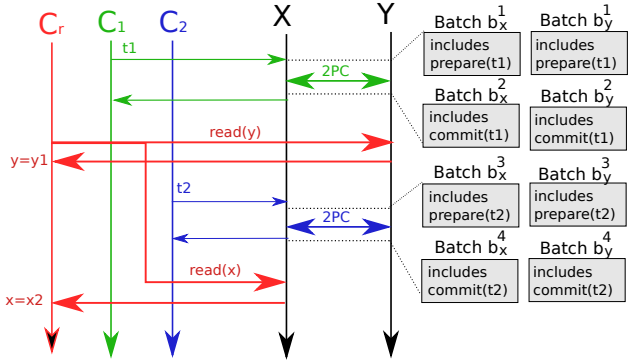}
\caption{A motivating example showing that simply doing local
read-only transactions could lead to inconsistent distributed
read-only transactions}
\label{fig:rot-motivate}
\end{figure}

TransEdge protocols require complex dependency tracking mechanisms because
using an ADS by itself is not sufficient to ensure data consistency across
partitions.
The example in Figure~\ref{fig:rot-motivate}
demonstrates that using Merkle Trees without additions might lead to
inconsistencies when reads are distributed across more than one
partition. This demonstration motivates our cross-partition
dependency tracking mechanism that we develop in the rest of this
section.  Consider \cut{, for this example,} a distributed read-only
transaction, $t_r$, that reads two data objects, $x$ and $y$, from
two partitions, $X$ and $Y$, respectively. Assume that concurrent to
$t_r$, there were two distributed read-write transactions, $t_1$ and
$t_2$. Both read-write transactions do the same thing---each writes a
new value for both $x$ and $y$. Transaction $t_1$ writes values $x_1$
and $y_1$ first, then transaction $t_2$ writes values $x_2$ and
$y_2$.

If transaction $t_r$ executes concurrent to $t_1$ and $t_2$,
then it must
return one of three snapshots to satisfy serializability: (1)~return the
initial snapshot where neither $t_1$ nor $t_2$ has committed, (2)
return $x_1$ and $y_1$, which correspond to the database state after
committing $t_1$, or (3) return $x_2$ and $y_2$, which correspond to
the database state after committing $t_2$. However, using Merkle
Trees without additions can lead to returning an inconsistent read,
such as returning the values $x_1$ and $y_2$.  For this to happen,
assume that $t_1$'s commit record was in batch 2 in $X$ and in batch
2 in $Y$. Also, assume that $t_2$'s commit record was in batch 4 in
$X$ and in batch 4 in $Y$. If $t_r$ started while $t_2$ is sending
the commit messages to participants, it is possible that $t_r$ reads the state of $X$ as of batch 4, while reading the state of $Y$
as of batch 2. This leads to an inconsistent read-only transaction.

We begin this paper with an overview and background of TransEdge
(Section~\ref{sec:background}). 
%
We propose TransEdge design and transaction processing protocols in
Section~\ref{sec:txn}, and propose TransEdge optimized snapshot read-only
transaction processing protocol in Section~\ref{sec:rot}.  An
experimental evaluation is presented in Section~\ref{sec:eval}.
Related work is presented in Section~\ref{sec:related} and
Section~\ref{sec:conclusion} concludes the paper.

%
%
%
%
%
%


\section{TransEdge System Model}\label{sec:background}

\textbf{System Model.}
Global-Edge Data Management (GEDM) is a data management model that
aims to utilize both cloud and edge resources. GEDM aims to overcome
wide-area latency by leveraging the resources at the edge of the
network. These edge resources can take the form of private, edge, and micro
datacenters which are clusters of machines. These infrastructures are typically
maintained at the edge of the network by third-party organizations
like Internet service providers.

\textbf{Security and trust model.}
We adopt a byzantine fault-tolerance
(BFT) model where a node might act in arbitrary and malicious ways.
We assume that the number of such malicious nodes is bounded by a
number $f$. We build upon prior work on BFT systems where clusters of
machines are used to mask $f$ failures by replicating and
coordinating across $3f+1$ nodes~\cite{castro1999correctness}.


\textbf{Data model.}
Data is partitioned across edge nodes and each partition is
replicated across clusters of $3f+1$ nodes, where $f$ is the number of tolerated
byzantine failures. We assume that a clustering protocol is utilized to form the partitions with such a guarantee, similar to prior work in hierarchical BFT protocols~\cite{gupta13resilientdb,nawab2019blockplane,amiri2019parblockchain,hellings2021byshard}. More details about node grouping is presented in Section~\ref{sub:related-byz}. Coordinating access and data operations is
performed via a hierarchical BFT architecture, where transactions
local to a partition are performed within the corresponding cluster.
These local transactions are committed via BFT protocols such as
BFT-SMaRt. Inter-cluster operations are performed via benign protocols
(\emph{e.g.}, Two-Phase Commit) that are layered over the BFT
replication layer~\cite{nawab2019blockplane,gupta13resilientdb,amir2010steward,amiri2019parblockchain,al2017chainspace}.
Due to spatial (geographic) locality, inter-cluster operations are typically across nearby clusters (e.g., it is typical that a user interacts with users who are geographically close to them which leads to data and operations to have such geographic locality).



\textbf{Interface.} A client running a transaction first creates a Transaction object and sends \textsf{read} operations to access data from TransEdge nodes. \textsf{write} operations are buffered at the client.
When the transaction is ready, a \textsf{commit} request groups read and write operations and sends them to a TransEdge node to be committed. TransEdge supports efficient snapshot read-only transactions with the
interface \textsf{read} and \textsf{commit-rot} which executes a special read-only transaction algorithm. Each edge node has
a unique public/private key that it uses in all communications with
the other edge nodes.

\cut{
TransEdge consists of a transaction processing layer and a replication
layer. The transaction processing layer receives read and transaction
begin and end commands. (Write operations are buffered at the
client.) The transaction processing layer batches transactions and
commits them to the BFT SMR log by calling a \textsf{commit} call to
the replication layer. In the rest of the paper, we provides the
implementation details of these layers: GBFT distributed consensus is
used to implement the replication layer (Section~\ref{sec:gbft}),
distributed transaction processing is performed by Two-Phase
Commit~(2PC) (Section~\ref{sec:txn}), and snapshot read-only transactions are
performed by utilizing byzantine-replicated authenticated data
structures and dependency tracking (Section~\ref{sec:rot}).
}

\section{Read-Write Transaction Processing}
\label{sec:txn}

\cut{In this section, we present the transaction processing protocols for
read-write transactions in TransEdge.} TransEdge's Read-write transaction
processing protocol facilitates dependency tracking
that we later use for efficient read-only transaction processing. We
only provide the
design of read-write transactions in this section and then present
the snapshot read-only transaction processing protocols in Section~\ref{sec:rot}.


\begin{figure*}[!t]
\begin{center}
\includegraphics[scale=1.0]{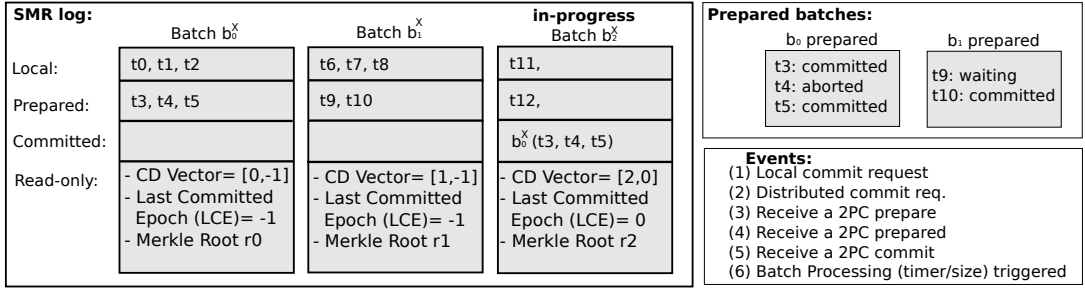}
\caption{An example of the state of a leader node.} 
\label{fig:state-2pc}
\end{center}
\end{figure*}

\subsection{System Model and Overview}
\label{sub:model}

Data is
partitioned into $n$ partitions. The state of each partition is
maintained via a BFT State-Machine Replication (SMR) log. This replication log is managed by a cluster of edge machines. Each cluster consists of $3f+1$ edge nodes where up to $f$ nodes can be malicious (i.e., byzantine). In our
case, the BFT SMR is implemented by BFT-SMaRt~\cite{bft-smart-lib}.

A leader writes data to the SMR log in batches. Batches are written
one-by-one, \emph{i.e.}, a leader writes a batch only if the previous
batch is already written. Each batch consists of multiple segments (Figure~\ref{fig:state-2pc}):
\begin{enumerate}[leftmargin=*]
\item Local transactions segment, which includes local
transactions. Local transactions are defined as transactions with read and write operations on keys that are all local to the cluster where they are being sent.
\item Prepared distributed transactions segment, which includes
distributed transactions that are Two-Phase Commit (2PC) prepared but not committed yet.
Distributed transactions are defined
as transactions whose read or write set contains keys that are not local to one cluster.
\item Committed distributed transactions segment, which includes
distributed transactions that are 2PC committed. 
\item Read-only segment, which includes (i)~a conflict-dependency
(CD) vector that tracks dependencies on transactions at other
partitions, (ii)~a Last Committed Epoch (LCE) number that represents the
largest batch number where prepared transactions are committed (the LCE also serves as a version number for the keys committed in the batch as we will discuss later)
 and (iii)~a Merkle Tree root which is used to certify the
integrity of committed transactions. The read-only segment is
updated and written while processing local and distributed
transactions, but their utility is in
enabling efficient snapshot read-only transactions as we show later
(Section~\ref{sec:rot}).
\end{enumerate}
A batch is represented with the notation $b^X_i$, which denotes the
$i^{th}$ batch in partition $X$. The subscript $i$ serves as a timestamp for each batch in the SMR log. Each segment is represented by
appending the segment's name to the batch---for example,
$b^X_i$.local represents local transactions in $b_i^X$.

An example of the state of a leader edge node is shown in
Figure~\ref{fig:state-2pc}. The SMR log contains two written batches
($b_0$ and $b_1$) and also shows the \emph{in-progress batch} ($b_2$)
that is being constructed to be written to the SMR Log. Batch $b_0$
contains three local transactions---$t_0$, $t_1$, and $t_2$. Local
transactions are considered committed as soon as the batch is
written to the SMR log. Batch $b_0$ also contains three distributed
transactions---$t_3$, $t_4$, and $t_5$---in the prepared segment.
Each transaction in the prepared segment is prepared as of batch
$b_0$---the term \emph{prepared} here corresponds to the prepared
state of 2PC which guarantees the property that no two conflicting
transactions can be prepared at the same time.) These prepared
transactions are waiting to hear from the other partitions/leaders to
know whether to commit or abort.

The leader tracks the
state of prepared batches in a data structure called \emph{prepared batches}
shown in Figure~\ref{fig:state-2pc}. Once all the transactions in the
oldest prepared batch are ready (committed or aborted), the leader
adds these transactions to the committed segment of the next batch
(the in-progress batch). In the figure, the prepared transactions of
$b_0$ are ready, and thus are added to the committed segment of batch
$b_2$. Prepared transactions are not considered committed until the
batch with the corresponding committed segment is written to the SMR
log. 

\cut{
TransEdge adopts an optimistic concurrency control model. A client
sends read requests to the leader, but writes are buffered. When the
client finishes all operations it sends a \emph{commit request} to
the leader. If the transaction is local (it only accesses data
objects belonging to a single partition), the leader commits the
transaction by writing them in the local segment of the next batch.
Then, it responds to the client with a proof of its commitment. The
proof is a collection of $f+1$ signed messages that the transaction
has committed. 

Distributed transactions are processed using a Two-Phase Commit (2PC)
layer on top of the SMR log. Each step of the 2PC protocol is
logged in the \textsf{prepared} and \textsf{committed} segments. (Proofs
from each partition are collected to mask byzantine failures in the
2PC process.) A transaction $t$ in the prepared segment of batch
$b_i$ means that the transaction is prepared as of $b_i$. Likewise, a
transaction is committed as of $b_j$ is it is in the committed
segment of that batch. For example, in Figure~\ref{fig:state-2pc},
transactions $t_3$, $t_4$, and $t_5$ are prepared as of batch $b_0$
and committed as of batch $b_2$.
}


\vspace{-0.1in}
\subsection{Intra-cluster Transaction processing}
\label{sub:local-rw}

\cut{The leader of a partition commits local transactions to the SMR log in the
\textsf{local} segment of batches, $b_i^X.$\textsf{local}.} TransEdge
processes transactions in an Optimistic Concurrency Control approach \cite{kung1981optimistic,schlageter1981optimistic}.
Clients' read requests are served from any node in the cluster. Responses to clients must include the LCE of the batch which the key was read from. The client buffers the writes in the transaction object. When it is time to commit the transaction, the client sends the transaction object containing the read and write sets to the leader. The leader
receives commit requests from clients that include the read and
write-sets (event 1 in
Figure~\ref{fig:2pc-scenario}). The read-set is the set of read data
objects along with the read values. The write-set is the set of data
objects to be written.  When a transaction, $t$, is received by
the leader, the leader verifies that it can commit and appends it to
the in-progress batch. This verification involves checking the
following conflict detection rules, which we will also use later for
distributed transactions:

\begin{definition}
\label{def:conflicts}
\textbf{(Conflict detection rules.)}
A transaction $t$ is added to an in-progress batch only if it does not
conflict with the following:
\begin{enumerate}[leftmargin=*]
  \item Previous batches: this is checked by verifying
that all reads in the read-set of $t$ are not overwritten by committed local
or distributed transactions in previous batches.
  \item In-progress batch: this is checked by verifying
that every transaction $t'$ in the local, prepared, or committed segments
do not conflict with $t$.
  \item Prepared transactions: this is checked by
verifying that every transaction $t'$ in prepared batches
(that are not committed) does not conflict with $t$.
\end{enumerate}
Otherwise, $t$ is aborted or restarted. 
\end{definition}

Transactions that pass the conflict detection rules are added to the in-progress batch of the log. The batch is replicated within the cluster by the leader. Depending on whether the transactions are local or distributed they are added to the specific segments 
of the current batch of the SMR log. The cluster replicas perform the same conflict detection on the transaction batch and then add the transactions to their SMR log. All cluster replicas form a Merkle tree of the local and distributed transactions once conflict detection is complete. A consensus process is performed on the merkle root of the transactions added to the batch using the BFT-SMaRt protocol\cite{bft-smart-lib}. At the end of the consensus $f+1$ signatures are collected from the replicas and are added to the batch. Once the consensus is complete the local transactions are considered committed to the batch. If the consensus fails, the transactions are aborted. 
Local (intra-cluster) transactions are considered committed after their batch is
written to the SMR log. Writing to the SMR log entails coordinating
with other replicas using the PBFT protocol. Other replicas, when
engaged in writing the batch to the SMR log, ensure that the local
transactions are in fact allowed to commit using the rules above
(Definition~\ref{def:conflicts}). This guarantees that a malicious
leader cannot commit transactions that are inconsistent with the state of the SMR log of other replicas in the cluster. Additionally, the batch generated by following the conflict definition in definition~\ref{def:conflicts} ensures that the transactions are conflict serializable. 


\cut{This segment of a log entry  consists of a \textsf{coordinator-prepare} record that keeps track of the prepared transactions in the cluster for any given batch.}

\begin{figure}[!t]
\begin{center}
\includegraphics[scale=.6]{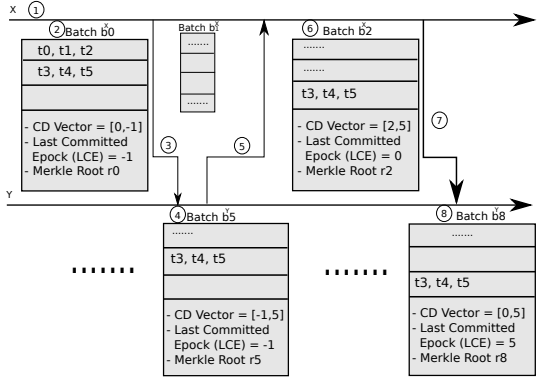}
\caption{An example of committing distributed transactions across two
partitions, X and Y.}
\label{fig:2pc-scenario}
\end{center}
\end{figure}

\begin{sloppypar}
\end{sloppypar}
\subsection{Distributed Transaction Processing}
\label{sub:dist-rw}

A distributed transaction accesses more than one partition. We use a
2PC-based approach to commit distributed transactions by having the
2PC protocol as a layer on top of the consensus protocol.
Specifically, each step of the 2PC protocol is verified and persisted
through the BFT protocol.  Prepared and committed transactions will be part of the prepared and committed segments, respectively (Figure~\ref{fig:state-2pc}.)
In the following, we present the distributed transaction processing
protocol. We use the example in Figure~\ref{fig:2pc-scenario} to
illustrate each step. 
\subsubsection{\textbf{Client protocol}}
\label{subsub:client}
In TransEdge, a client creates a Transaction Object that reads data from a corresponding partition and buffers writes. Read or write
objects might belong to more than one partition. When the client
is ready to commit,
it picks one of the clusters of the accessed
partitions to act as a whole as the \emph{coordinator} of the transaction. Then,
it sends a commit request to the coordinator containing the read and write set (step 1 in
Figure~\ref{fig:2pc-scenario}.) The coordinator then
drives the commitment of the transaction. Each step that is performed by the coordinator cluster is verified using the underlying BFT protocol. This ensures that a malicious leader or node in the coordinator cluster would not be able to lie when communicating 2PC steps with other accessed clusters. Likewise, the other clusters participating in 2PC also verify all 2PC steps with the underlying BFT protocol.

When the client communicates with the coordinator cluster, it can send the request to $f+1$ nodes in the cluster to ensure that malicious nodes would not drop the message. Once the request is written in the BFT cluster, the next steps will be driven by the cluster as a whole which will prevent a malicious leader or node from hindering the progress of the 2PC protocol. Communication between clusters for 2PC steps can also adopt a similar strategy by making $f+1$ nodes send relevant 2PC messages to $f+1$ nodes in other clusters for each step.
%

\subsubsection{\textbf{2PC prepare}}
When the coordinator receives a commit request for a transaction $t$,
it verifies that $t$ can prepare. This is the 2PC prepare phase and is described in section~\ref{sub:local-rw}. It is shown in step 2 in Figure~\ref{fig:2pc-scenario}.
\cut{This is done by checking the
conflict detection rules (Definition~\ref{def:conflicts}). Then, the
leader constructs a \textsf{coordinator-prepare} record that includes
the transaction's information (read and write sets and the request
received from the client). During this 2PC prepare phase, the coordinator
prepares keys local to its partition described as follows. The \textsf{coordinator-prepare} is added
to the \textsf{prepared} segment of the in-progress batch, which will
eventually be written to the SMR log (step 2 in
Figure~\ref{fig:2pc-scenario}.)}
After the transaction is written to the SMR log,
the leader sends the \textsf{coordinator-prepare} with a proof that
it is part of the SMR log ($f+1$ signatures collected during \textsf{coordinator-prepare}) to the leaders of the
accessed partitions (step 3 in Figure~\ref{fig:2pc-scenario}.) 

\subsubsection{\textbf{2PC prepared}}
When a leader receives a \textsf{coordinator-prepare} for a
distributed transaction, it executes the intra-cluster transaction processing protocol described in section \ref{sub:local-rw}.
It then constructs a \textsf{prepared} record and adds it to the \textsf{prepared}
segment of the next batch (step 4 in Figure~\ref{fig:2pc-scenario}.)
Each replica, while writing the batch, verifies the authenticity of
the \textsf{prepare} record and verifies that it can commit using
Definition~\ref{def:conflicts}. Once the batch with the
\textsf{prepared} record is written, the leader sends it to the
coordinator (step 5 in Figure~\ref{fig:2pc-scenario}.) The message
includes the \textsf{prepared} record signed by $f+1$ nodes in the
partition. The set of prepared transactions is also added to the prepared batches data structure.  

\subsubsection{\textbf{2PC commit at the coordinator}}
When the coordinator receives \textsf{prepared} messages from all the
participating partitions, it declares whether the transaction commits
or aborts. If all \textsf{prepared} messages are positive, then the
transaction commits. Otherwise, it aborts. The
coordinator constructs a \textsf{commit} record that includes the
collected \textsf{prepared} messages. The coordinator, then, writes
the \textsf{commit} record in the prepared batches structure for its
corresponding batch $b_p$
(Figure~\ref{fig:state-2pc}.) Once all the other transactions in
$b_p$ are ready, they are added to the committed segment of the
next in-progress batch $b_i$ (step 6 in Figure~\ref{fig:2pc-scenario}.)
When $b_i$ is written to the SMR log, the distributed transaction is
considered committed. The leader then updates the LCE segment of the batch.  Afterward, the leader sends the commit
record---along with $f+1$ signatures---to the other partitions'
leaders (that were accessed by the transactions) and to the client
(step 7 in Figure~\ref{fig:2pc-scenario}.)  When a leader receives a
\textsf{commit} record, it updates the corresponding batch in the
prepared batches data structure. Once the corresponding prepared batch is ready and is the next one to be committed, it is added in the committed segment of the next batch and committed to the SMR log (step 8 in
Figure~\ref{fig:2pc-scenario}.) In the next section, we provide the updates that need to be applied to the Read-Only part of the batch as part of commitment.

\subsubsection{2PC transactions across more than two clusters}
\label{subsub:nclusters}
The distributed transaction processing protocol above can be performed with more than two clusters. To illustrate, consider extending the scenario shown in Figure~\ref{fig:2pc-scenario} with one more cluster/partition called Z. If the distributed 2PC transactions in $B_0^X$ access records in Z as well, then the changes to the scenario are as the following. The \textsf{coordinator-prepare} message (step 3) and the commit record (step 7) is sent to both $Y$ and $Z$. Also, the transactions prepare (steps 4 and 5) and commit (steps 7 and 8) are performed at both $Y$ and $Z$. 

Another aspect of having more than two clusters is that distributed 2PC transactions may start at different clusters. Consider the previous setup in Figure~\ref{fig:2pc-scenario} with an additional cluster $Z$. A distributed 2PC transaction $t_6$ that accessed both $Z$ and $Y$ may start at $Z$ as the coordinator concurrently with transactions $t_3$, $t_4$, and $t_5$. In this case, $t_6$ is first prepared in a batch at $Z$. Then, a \textsf{coordinator-prepare} is sent to $Y$. When $Y$ receives the message, it decides whether it can be committed. For example, if the message for $t_6$ is received after receiving the message for $t_2$, $t_3$, and $t_4$, cluster $Y$ checks for conflicts including conflicts between $t_6$ and the prepared transactions $t_2$, $t_3$, and $t_4$. Then, processing $t_6$ continues by sending back the \textsf{prepared} message, and then for $Z$ to send back the \textsf{commit} message.

\subsubsection{Transaction Aborts}
\label{subsub:aborts}

{In TransEdge, aborts may be caused due to various factors. First, a transaction may be aborted due to conflicts. During conflict checks, using the conflict detection rules we outlined above, if a transaction $t$ has a conflict, then it is marked as aborted. Transaction $t$'s write-set is not applied to storage and it does not impact future transactions. The user client may request to abort a transaction while it is processing. However, once the commit request is sent to a cluster, a user client cannot request for the transaction to be aborted. Also, TransEdge does not abort transactions due to a timeout. Because each 2PC participant is represented as a cluster of machines, we assume that there is always a quorum of nodes in each cluster that can make progress, and therefore, an abort due to timeout is unnecessary.

\subsection{Commit Updates in the Read-Only Segment}
\label{sub:batch}
At the time of committing a batch of transactions (step 8 in
Figure~\ref{fig:2pc-scenario}), the following updates need to be applied to the Read-Only part of the batch (these updates---although not impacting the processing of read-write transactions---affect the processing of read-only transactions that we present in Section~\ref{sec:rot}. We present more details about computing these updates in Section~\ref{sec:rot}):
\cut{
Each leader maintains the following threads that keep track of each stage of transaction processing:
\begin{enumerate}
    \item \emph{batching worker}. Receives transactions and initiates intra-cluster transaction processing (section \ref{sub:local-rw}
    \item \emph{2pc-prepare worker}. Sends 2PC prepare requests to remote cluster leaders. Updates a \textsf{prepared-batches} data structure to keep track of the sequence of prepare messages by using the batch ID.
    \item \emph{2pc-prepared worker}. Responds to prepared messages from remote cluster leaders. It updates \textsf{prepared-batches} data structure to keep track of prepared messages for specific batch ID.
    \item \emph{2pc-commit worker}. Ensures that batches are committed sequentially. For eg. even if all of batch $b_j$'s transactions were prepared by remote leaders, $b_j$ will not be able to commit unless all batches $b_i$ where $i<j$ are committed.
\end{enumerate}
}
\cut{(called the \emph{batching thread})
that takes care of constructing batches and then driving their
commitment to the SMR log.  Initially, the in-progress batch is empty. Threads
receiving the transactions (called \emph{worker
threads}) add the transactions to the in-progress batch by sending
them to the batching thread. These include transactions in the local
and prepared segments.}
\cut{
A new transaction is added into a batch only if it satisfies conditions specified in definition \ref{def:conflicts} which ensures that a transaction added to a batch is conflict serializable. Secondly, the transaction is serializable with the transactions in the existing in-progress batch. And finally, that there are no existing locks on any of its objects. This ensures that batch processing does not get into deadlocks and guarantees serializability.}
\cut{Periodically, the batch processing thread wraps up the
construction of the batch and then commits it via the BFT protocol.
Completing the construction of the batch includes updating the following:}
\begin{enumerate}[leftmargin=*]
  \item \emph{Committed segment and Last Committed Epoch (LCE): }
The cluster replica observes the \textsf{prepared-batches} data
structure (Figure~\ref{fig:state-2pc}). Specifically, it checks
whether the earliest batch, $b_i$, is ready---it has no pending
transactions. If $b_i$ is ready, the transactions in it are added to
the committed segment of the in-progress batch, the Last Committed
Epoch (LCE) number is updated to $i$
and $b_i$ is removed from the prepared batches data structure.
The LCE represents the id of the most recent batch that committed as of the current in-progress batch. 
  \item \emph{Conflict Dependency (CD) Vector: } The CD vector
encodes information about the dependencies from the committed batch
to the batches of other partitions/clusters. The CD vector at a cluster
contains $n$ entries, where $n$ is the number of partitions. Each partition is
represented by a number in the CD vector. For example, if the CD
vector contains the number $i$ for partition $X$, this means that
there is a dependency on transactions at $X$ up to batch $b_i^X$. An entry in the CD Vector can track the dependencies of multiple 2PC transactions by taking a coarse-granularity approach. 2PC transactions in a batch are all reading from the same data state that represents the committed transactions of the received batches up until that point. This makes all the transactions in a batch have the same set of potential dependencies to the state represented by the received batches up until that point. More information about dependency tracking and the representation of dependencies is presented in Section~\ref{sub:rot-dist}. 
  \item \emph{Merkle root:} The Merkle Tree \cut{---which is used for
snapshot read-only transactions and will be introduced in more detail in
Section~\ref{sec:rot}---}is updated with the write-sets of the transactions in the local, prepared and committed segments. \cut{Then, the new Merkle Tree root is added to the batch.}
This Merkle Tree root represents the state of the Merkle Tree that
includes all local, prepared and committed transactions up to the current
batch. The Merkle tree is updated by all replicas within a cluster while processing read-write transactions. 
\end{enumerate}

\subsection{Comparison with Hierarchical 2PC/BFT}
\label{sub:compare}

TransEdge extends the literature of hierarchical 2PC/BFT systems, which are systems that perform 2PC across clusters, where each cluster maintains a shard of the data~\cite{gupta13resilientdb,al2017chainspace} (similar to these systems are 2PC/Paxos systems that utilize paxos as the underlying consensus layer~\cite{spanner}). Specifically, we consider a baseline inspired by these hierarchical 2PC/BFT systems that we call \emph{2PC/BFT} and use in our evaluation study. In 2PC/BFT, similar to TransEdge, each cluster acts as a 2PC participant and each step taken by the participant cluster is first validated by the underlying BFT consensus/agreement protocol. The main differentiator of TransEdge is the efficient support of read-only transactions. This leads to the main design differences such as having to maintain the Read-only segment of batches, maintaining a CD Vector, a LCE number, and a merkle tree (Figure~\ref{fig:state-2pc}). In terms of algorithms, what TransEdge has additional support to maintain these added structures (such as CD Vectors), new algorithms for read-only transactions that we discuss in Section~\ref{sec:rot}, as well as constraints on how distributed read-write transactions are performed to enable efficient read-write transactions (Section~\ref{subsub:dependency}). 2PC/BFT systems, on the other hand, do not maintain these additional structures, do not have additional constraints on the original 2PC protocol, and do not have additional special algorithms for read-only transactions. However, they do not have the efficient support of read-only transactions that TransEdge provides. We discuss further comparisons with related work in Section~\ref{sec:related}.

2PC/BFT systems are used as a solution to the problem of performing geo-distributed coordination across byzantine nodes that are distributed around large geographic regions. 2PC/BFT groups nodes together into clusters, each maintaining a shard. Then, a cluster can act independently on behalf of that shard of data. This includes coordinating with other shards which can be done via the 2PC protocol. Therefore, the 2PC/BFT protocol enables clustering data into smaller groups of nodes that are near each other, and then utilizes an intra-cluster BFT instance to validate steps that can be taken by the cluster as a whole. This enables using benign (non-byzantine) algorithms across clusters since each step is validated internally. This led some prior work to adopt this paradigm, such as ResilientDB~\cite{gupta13resilientdb} and ChainSpace~\cite{al2017chainspace}. This concept of layering 2PC over consensus is also used in systems such as Spanner~\cite{spanner} which utilize paxos as the consensus layer since they do not consider byzantine failures.

\subsection{Correctness}
\label{sub:rw-correctness}

TransEdge guarantees serializability~\cite{bernstein1987concurrency}. We present now a proof sketch of this guarantee. 

We define \emph{transaction commit points} (TCP). A transaction commit point for a distributed transaction $t$ is the time it commits (i.e., the coordinator node has received all positive votes). For a local transaction $t$, the transaction commit time the batch containing $t$ has been written to the SMR log (i.e., the leader received enough votes for its commitment.)
We denote the transaction commit time for a transaction $t$ as $TCP(t)$.

Second, conflicts between transactions are defined as follows~\cite{bernstein1987concurrency}: (1)~write-read (wr) conflicts: a wr conflict from transaction $t_1$ to $t_2$ exists if transaction $t_1$ has a write operation on an object $o$ and $t_2$ reads the object version written by $t_1$. (2)~read-write (rw) conflicts: a rw conflict from transaction $t_1$ to $t_2$ exists if transaction $t_1$ reads an object $o$ and $t_2$ overwrites the object version read by $t_1$. (3)~write-write (ww) conflicts: a ww conflict from transaction $t_1$ to $t_2$ exists if transaction $t_1$ writes to an object $o$ and $t_2$ overwrites the object version written by $t_1$.

\begin{lemma}
\label{lem:equal}
For any two transactions $t_1$ and $t_2$, where $TCP(t_1) = TCP(t_2)$, it is guaranteed that $t_1$ and $t_2$ do not conflict with each other. 
\end{lemma}
\begin{proof}
There are two cases:\\
(1)~both $t_1$ and $t_2$ are local transactions: if $t_1$ and $t_2$ are at the same shard (i.e., data partition), then they cannot conflict as having the same TCP implies they are both at the same batch. If they are in the same batch, the two transactions do not conflict as we do not include conflicting transactions in the same batch. If $t_1$ and $t_2$ are at different partitions, then they do not conflict as they do not have any operations in common.\\

(2)~One of the two transactions is distributed: assume that $t_1$ is a distributed transaction and $t_2$ can be either local or distributed. At the TCP of $t_1$, all the shards accessed by $t_1$ are already prepared (i.e., a prepared message is already received). This means that all prepared shards would not prepare or commit another transaction---local or distributed---before finalizing the commitment of $t_1$. If $t_2$ conflicts with $t_1$, then it would be in one of the prepared shards of $t_1$. This is a contradiction as $t_2$ cannot commit at the same TCP of $t_1$ in a $t_1$-prepared shard.

\end{proof}

\begin{lemma}
\label{lem:order}
For any two conflicting transactions $t_1$ and $t_2$, where there is a data conflict (i.e., read-write, write-read, and write-write~\cite{bernstein1987concurrency}) from $t_1$ to $t_2$, it is guaranteed that $TCP(t_1) < TCP(t_2)$
\end{lemma}
\begin{proof}
Assume to the contrary that there is a conflict from $t_2$ to $t_1$ while $TCP(t_1) < TCP(t_2)$:\\
(1)~$t_2 \rightarrow_{wr} t_1$: If $t_1$ reads an object $o$ from $t_2$, then this implies that at the time of the read operation $time(r_1(o))$ the value being read is already committed. Therefore, $TCP(t_2) < time(r_1(o))$. Also, $t_1$ TCP is after performing all the read operations meaning that $time(r_1(o)) < TCP(t_1)$, which implies $ TCP(t_2) < TCP(t_1)$, which is a contradiction.\\

(2)~$t_2 \rightarrow_{ww} t_1$: If $t_1$ overwrites an object $o$ written by $t_2$, then this implies that at the time of the write operation $time(w_1(o))$ the value being overwritten is already committed. Therefore, $TCP(t_2) < time(w_1(o))$. Also, $t_1$ TCP is the same as the time of applying the write operations meaning that $time(r_1(o)) = TCP(t_1)$, which implies $ TCP(t_2) < TCP(t_1)$, which is a contradiction.\\

(3)~$t_2 \rightarrow_{rw} t_1$: If $t_1$ overwrites an object $o$ read by $t_2$, then either one of the following cases would apply: (3.a)~$t_1$ prepares before $t_2$: this is not possible because at the time $t_2$ attempts to prepare at the shard containing $o$, it will discover the conflict with $t_1$ and aborts. (3.b)~$t_2$ prepares before $t_1$: in this case, for $t_1$ to not abort, it needs to prepare after the commitment of $t_2$. However, in that case, $t_2$'s TCP is before $t_1$'s TCP, $TCP(t_2) < TCP(t_1)$, which is a contradiction.
\end{proof}

\begin{theorem}
TransEdge read-write transactions guarantee serializability.
\end{theorem}
\begin{proof}
We utilize the serializability graph (SG) test. The SG is defined as a graph where each transaction is represented as a vertex and each conflict is represented as a directed edge between two transaction vertices. The SG test is the following: if the SG contains no cycles, then the transaction history is serializable~\cite{bernstein1987concurrency}. We now demonstrate that TransEdge ensures that the SG does not contain cycles for any transaction history.

Assume to the contrary that there is a cycle in the SG, $t_1 \rightarrow \ldots \rightarrow t_n \rightarrow t_1$. Using Lemmas~\ref{lem:equal} and \ref{lem:order}, we derive that $TCP(t_1) < TCP(t_2)$. By transitivity, we arrive at $TCP(t_1) < TCP(t_1)$, which is a contradiction. This proves the serializability of any transaction history generated by TransEdge. 
\end{proof}

\section{Read-only Transactions}
\label{sec:rot}

TransEdge \cut{ is the first database to} supports efficient serializable \emph{read-only transactions} in a byzantine environment.  The advantages
of an efficient snapshot read-only transaction are: (1) commit-freedom: the
transaction needs to coordinate with only a single node from each
accessed partition without involving the other replicas, and (2)~non-interference: the transaction does not interfere with on-going
read-write transactions.  The first feature reduces the cost of
read-only transactions significantly since the alternative is to
commit read-only transactions as regular read-write transactions,
which incurs a coordination cost across at least $2f+1$ nodes per accessed
partition and running 2PC and BFT protocols that incur significant overhead (Section~\ref{sub:dist-rw}.) 
The
second feature (non-interference) means that read-write transactions
can proceed even in the presence of conflicting read-only
transactions. 
\subsection{Overview}
\label{sub:rot-overview}

%
In TransEdge, we extend the use of Merkle Trees~\cite{merkle1980protocols} for database applications~\cite{jain2013trustworthy}
in two ways that are not supported in earlier systems:
\begin{enumerate}
\item Support distributed read-only transactions across multiple
untrusted nodes or clusters of nodes.
\item Support updating the Merkle Tree without relying on the
involvement of a centralized trusted entity.
\end{enumerate}

When a leader commits a
new batch of transactions, it recomputes the Merkle Tree to reflect
the new data. To ensure that the new recomputed Merkle Tree is
authentic, the leader commits the root of the new Merkle Tree with
the corresponding batch. Replicas verify the authenticity of the new
root and provide a signed message of the new root. When a client
sends a local read-only transaction request, the leader responds with
the corresponding data blocks, the sibling nodes in the path to root
for all blocks, and the root of the tree with $f+1$ signatures.

Distributed read-only transactions are performed by querying all accessed partitions. However, this is not sufficient because the
responses of the different leaders might be inconsistent---unlike the
local read-only transaction case, where we know the state is
consistent because they all belong to one Merkle Tree which
represents a consistent snapshot.  We propose a multi-round
distributed read-only transaction protocol that ensures reading from
a consistent snapshot across partitions. The protocol relies on
tracking dependencies across partitions, detecting inconsistencies in the first round, and satisfying any missing dependencies---if any---in the following round.



\subsection{Local Read-Only Transactions}
\label{sub:rot-local}
\cut{  
Each batch in the SMR log contains a Merkle tree root
(Figure~\ref{fig:state-2pc}.) The Merkle tree root corresponds to the
Merkle tree of all the local and committed transactions up to that
batch (\emph{i.e.}, prepared transactions are not going to affect the
state of the Merkle tree.) This Merkle tree root is used
for both local and distributed read-only transactions (distributed read-only transactions also use the rest
of the read-only segment of batches including the CD vector and the
LCE number.)}

The
client sends a \textsf{read-only-txn} request to the leader. The leader responds with the
requested data objects, sibling nodes in the path to the root for all
data objects, the Merkle Tree root, and $f+1$ signed messages proving
the authenticity of the root.
When the client receives the response, it verifies the authenticity
of the response and Merkle Tree root. Then, it computes
the Merkle Tree root using the received data blocks and sibling
nodes. If the computed root matches the received one (which is signed by $f+1$ nodes), then the client
accepts the received data. 



\subsection{Distributed Read-Only Transactions}
\label{sub:rot-dist}

Distributed read-only transactions are ones that read from multiple
partitions. TransEdge's support for distributed read-only transactions
builds on the support for local read-only transactions. The Merkle
Trees are utilized to generate proofs to authenticate data. However, we need to augment these
Merkle Trees with a dependency tracking mechanism to ensure that the
results are consistent \emph{across} partitions (see the example in
Figure~\ref{fig:rot-motivate}).

\subsubsection{\textbf{Overview and Intuition}}
To overcome inconsistencies between reads from different partitions,
we augment TransEdge with a \emph{dependency tracking} mechanism. This
dependency tracking mechanism enables detecting inconsistencies
of reads across partitions. Also, it enables identifying the needed
dependencies which allow TransEdge to ask partitions for missing
dependencies in following rounds. The read-only
algorithms utilize the information in the read-only segment of
batches. 

The dependency tracking mechanism relies on encoding the dependency
from one partition to all other partitions. Specifically, when a new
batch is committed, it includes all the dependencies from that
batch to other partitions. For example, in the previous motivating
scenario (Figure~\ref{fig:rot-motivate}), batch 2 in $X$, $b_2^X$,
would include the dependency information that corresponds to $t_1$. In
this case, $t_1$ is a distributed transaction that commits in both
$X$ ($b_2^X$) and $Y$ ($b_2^Y$). Therefore, $b_2^X$ has a dependency
relation to transactions in $b_2^Y$.

This dependency information allows detecting inconsistent reads.
For example, in the previous scenario the read-only transaction reads
the state of batch 4 at $X$, $b_4^X$, and batch 2 at $Y$, $b_2^Y$
(Figure~\ref{fig:rot-motivate}). In this case, a transaction committed in $b_4^X$ 
has a dependency to a transaction committed in $b_4^Y$.  When a client sees this dependency, it detect
an inconsistency since it read the state $b_2^Y$, which is earlier
than the required dependency, $b_4^Y$ 
(the dependency is satisfied if
the read batch is equal to or higher than the dependency.) 
In this
case, the client sends a request to partition $Y$ to get the state
that it depends on, which is $b_4^Y$ in this case.
 

\label{def:challenges}
\subsubsection{\textbf{Dependency tracking challenges 
}}
\label{subsub:dependency}
\begin{enumerate}[leftmargin=*]
\item Dependency representation granularity:
tracking
dependencies at the level of
transactions for all partition incurs a high overhead.
Our protocol tracks and detects dependencies at the granularity of
partitions instead of transactions, reducing the space of
dependencies to be in the order of the number of partitions rather
than the number of transactions. Also, we propose a
method to represent dependencies in a coarse way that allows
representing the dependencies of a large number of partitions with a
single number.
\item Allowing unconstrained local transaction commitment:
A feature that we want to maintain in TransEdge is the ability to
commit batches with local transactions without waiting until the
prepared transactions of the previous batch are ready. Therefore, if
transactions prepare in batch $b_i$, we cannot predict at which
future batch they will commit. This makes encoding dependencies
difficult, since we do not know which batch we are going to depend on
before others commit. To overcome this challenge, we encode
dependencies according to the batch where they were prepared, and
utilize the LCE number---which encodes the number of the batch when the
committed transactions were prepared---to check dependencies.

\end{enumerate}

\subsubsection{\textbf{Recording Dependency Information During
Read-Write Transaction Processing}}
\label{subsub:dependency}
In this section, we present how read-write transaction processing is
augmented to record dependencies that are used by read-only transactions. These changes are
essential to support the algorithms for read-only transactions that
we present in the following section.

\textbf{(a) Ordering Constraint on Distributed Read-Write Transactions.}
To enable efficient tracking of dependencies, we add a constraint on
the order of committing distributed read-write transactions at each
partition. Specifically, we group distributed read-write transactions
according to which batch contains their \textsf{prepare} records and
call them a \emph{prepare group}. We force transactions in a prepare
group to commit
together in a future batch (all transactions in the prepare group
commit together in one batch.) The ordering constraint is the
following:
\begin{definition}
\label{def:ordering_constraint}
The TransEdge \textbf{ordering constraint} forces prepare groups to
commit or abort in order; transactions in a prepare group in batch $i$ commit
before transactions in another prepare group commit or abort in batch $j$ if and only
if $i<j$.
\end{definition}
For example, consider two prepare groups, one denoted $P_1$ that is
part of batch $b_1^X$ and another denoted $P_2$ that is part of
$b_2^X$. Assume transactions in $P_1$ commit in batch $b_i^X$ and
transactions in $P_2$ commit in batch $b_j^X$. The ordering constraint
enforces that $i$ is less than $j$.

This ordering constraint enforces an order of when transactions can be committed. However, it still allows for concurrent processing in various ways. First, local transactions that do not conflict with in-progress batches can still be committed while in-progress batches of 2PC transactions are being processed and waiting to be committed. Second, distributed 2PC transactions can always concurrently be processed as long as they do not conflict with in-progress batches (and previously committed transactions). This is because the initial processing of transactions is performed by the client before the commit request time. After the transaction is ready, the commit request is sent, and the transaction enters the prepare phase. The processing of the prepare phase is not constrained by other batches. Third, distributed 2PC transactions in partitions that do not conflict with others can be performed concurrently. This is because two partitions with no conflicting transactions (direct or transitive) do not need to communicate with each other, and there is no ordering constraints that are enforced for them.

Although the ordering constraint enforces an order of commitment across batches, this is performed for transactions that already finished processing and requested to commit. Therefore, a long-running transaction is not going to increase the wait time due to the ordering constraint because it is only exposed to TransEdge after it is done processing and requested to commit. However, long-running transactions are not ideal in Transedge, since they may accumulate more conflicts before the request to commit phase and thus leading to a high chance of aborting due to conflicts (this is typical of occ-based solutions in general). As for long-running read-only transactions (such as scans and large analytics queries), these will not lead to performance degradation due to the TransEdge read-only protocol that makes them not conflict with ongoing read-write transactions.


\cut{
\begin{figure}[t]
\centering
\includegraphics[scale=0.77]{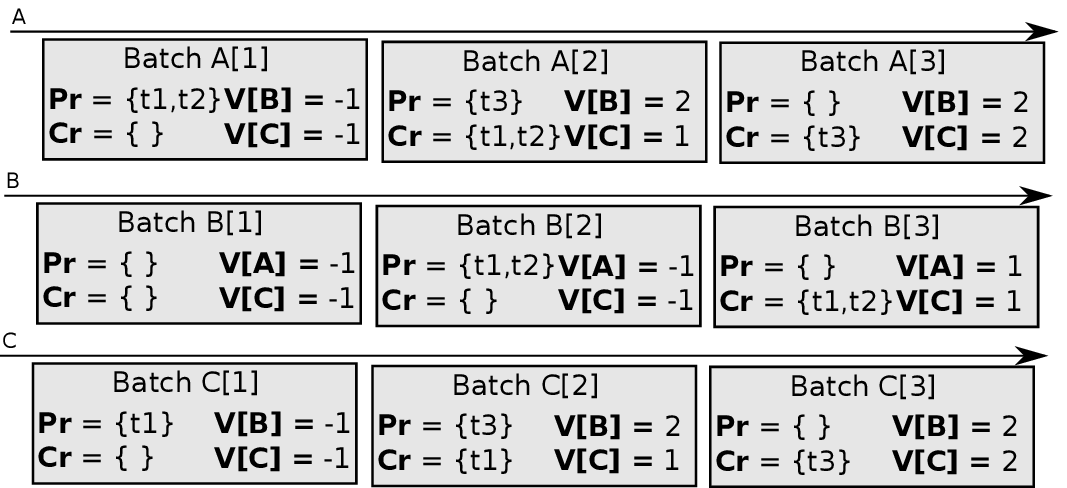}
\caption{An example of dependency vectors in a scenario with three
edge nodes and three transactions.}
\label{fig:dep-vectors}
\end{figure}
}

\textbf{(b) Dependency Tracking with CD Vectors.}
Given the ordering constraint, we are now able to represent
dependencies across partitions efficiently---all dependencies from
one batch to another partition in TransEdge are represented by one
number.  Specifically, each partition maintains a vector of
dependencies for every batch called the \emph{Conflict Dependency
(CD) vector}\cut{\footnote{The CD vector is an adaptation of vector
clocks to represent the dependencies from one partition to others.
However, unlike typical usage of vector clocks, each node's logical
clock does not correspond to communication events. Rather, it
corresponds to a batch.}}. We use the notation $\mathcal{V}_i^X$ to
denote the CD vector in batch $b_i^X$. The entry $\mathcal{V}_i^X[Y]$
denotes the dependency from batch $b_i^X$ to partition $Y$.
Specifically, the entry is the batch number at $Y$ that $X_i$ depends
on. 

The numbers in the CD vector represent the dependency to the batch
that contains the prepare records rather than the commit record. This
is because we want to allow unconstrained local transaction
commitment (Challenge 2 in subsection~\ref{def:challenges}.)
To allow unconstrained local transaction commitment, we cannot
predict (or enforce) at which batch a transaction would
commit at other partitions. However, we know at which batch the
transaction prepared. In other words, tracking using the batch where
transactions prepared complicates the design, but it enables us to
make local transactions commit with arbitrary frequency without delay. 

For example, consider the scenario in Figure~\ref{fig:2pc-scenario}.
The distributed transactions ($t_3$, $t_4$, and $t_5$) are prepared
in batches $b_0^X$ and $b_5^Y$, and committed in batches $b_2^X$ and
$b_8^Y$. (assume for ease of exposition that these are the only distributed transactions in
the scenario.) The CD vectors in the figure represent the dependency
to $X$ followed by the dependency to $Y$, \emph{i.e.}, $[i,j]$ represent
a dependency to $b_i^X$ and $b_j^Y$. The dependency from a batch to
the node writing the batch is always the batch id, \emph{i.e.}, the
CD vector $V_i^X$ has the value $i$ as the dependency to $X$.
Intuitively, this is because the state of the batch depends on all
the local and committed distributed transactions up to that batch. 

The dependency number to other partitions represents the batch in which
any common distributed transactions have prepared. For example,
observe the CD vector $V_2^X$ in batch $b_2^X$ of
Figure~\ref{fig:2pc-scenario}. In that batch, the distributed
transactions $t_3$, $t_4$, and $t_5$ committed. These transactions
were prepared in $Y$ at batch $b_5^Y$. This makes the dependency
number from batch $b_2^X$ to $Y$ be 5. Likewise, the dependency
number from batch $b_8^Y$ (where the distributed transactions commit
in $Y$) to $X$ is 0 (where the distributed transactions prepared in
$X$.)

Note in Figure~\ref{fig:2pc-scenario} the dependency values are initially -1 to represent the
absence of dependencies. The dependency relationship to
other partitions are only affected by the distributed transactions
that commit in the batch (not the ones that prepare.) It is possible
that the batch depends on multiple batches in another partition. In
this case, the dependency is to the latest batch of the multiple
dependencies.

\textbf{(c) Reporting Dependencies in Prepared Messages.}
For the batch processing thread to be able to derive dependencies for
the CD vector, it
needs to have the dependency information related to all transactions
in the committed segment. Specifically, what is needed are the batch
numbers
where the committed transactions prepared on other clusters. For
example, in Figure~\ref{fig:2pc-scenario}, while constructing batch
$b_2^X$, the leader of $X$ needs to know that the distributed
transactions $t_3$, $t_4$, and $t_5$ were prepared in $Y$ at batch
$b_5^Y$. The number of the batch will enable knowing the dependency
from $X$ to $Y$ in relation to the committed transactions.
Additionally, we need to know any transitive dependencies from the
batch that we depend on. For example, in
Figure~\ref{fig:2pc-scenario}, since the distributed transactions
lead to a dependency from $b_2^X$ to $b_5^Y$, the consequence is that
$b_2^X$ (transitively) depends on anything that $b_5^y$ depends on.

To summarize, the batch processing thread needs the direct and
transitive dependencies of all transactions in the committed segment.
To collect this dependency information, each \textsf{prepared}
message of a distributed transaction $t_i$ is piggybacked with the CD
vector of the batch $b_j^Y$ where $t_i$ is prepared. This piggybacked
CD vector encodes both the direct and transitive dependencies.

\alglanguage{pseudocode}
\begin{algorithm2e}[t]
\begin{small}
\caption{ Algorithm to derive dependencies to be
part of a new batch at partition $X$. }
\label{alg:derive}

\begin{algorithmic}[1]
\State $V^X$ := set of all dependency vectors at $X$
\State on event DeriveDepVector (in: i) \{ // i is the batch
number
\State \,\,\,\,\,$V_i^X \leftarrow V_{i-1}^X$
\State \,\,\,\,\,\textbf{for} commit record $cr$ in $b_i^X$.committed
\textbf{do}
\State \,\,\,\,\,\,\,\,\,\textbf{for} reported CD vector $V_j^Y$ in $cr$
\textbf{do}
\State \,\,\,\,\,\,\,\,\,\,\,\,\,\,$V_i^X =$ pairwise\_max($V_i^X$,
$V_j^Y$)
\State \}
\end{algorithmic}
\end{small}
\end{algorithm2e}

\textbf{(d) Deriving The CD Vector in Batches.}
Dependencies need to be derived while wrapping up the construction of
the in-progress batch (Section~\ref{sub:batch}). At a high-level, the
batch processing thread needs to go through all the transactions in
the \textsf{committed} segment in the batch and derive the
dependencies to other partitions according to the reported dependency
vectors in \textsf{prepared} messages.

The algorithm (Algorithm~\ref{alg:derive}) to derive the CD vector
$V_i^X$ starts by loading the CD vector of the previous batch,
$V_{i-1}^X$.  
Then, for each \textsf{commit} record of a transaction in the
committed segment ($cr$ in the algorithm), the leader processes all
the corresponding CD vectors. For example, consider a transaction $t$
that span three partitions, $X$, $Y$, and $Z$. When $X$ is deriving
the CD vector of the batch where $t$ commits, it uses the CD
vectors received in the \textsf{prepared} messages from $Y$ and $Z$.
For every reported CD vector $V_j^Y$, the algorithm performs a
pairwise maximum operation with the current $V_i^X$.  Eventually, the
new CD vector $V_i^X$ will be equal to the pairwise maximum of the CD
vector of the previous batch and all the reported CD vectors of
transactions in the committed segment. This new value of the CD
vector represents all the direct and transitive dependencies resulting
from committing the transactions in the committed segment.

\alglanguage{pseudocode}
\begin{algorithm2e}[t]
\begin{small}
\caption{Algorithm to verify dependencies in a
distributed read-only transaction}
\label{alg:verify}

\begin{algorithmic}[1]
\State on event VerifyDependencies (in: V) \{
\State // V is the set of received dependency vectors from
accessed partitions
\State \,\,\,\,\,\textbf{for} partition $i$ in accessed partitions
\textbf{do}
\State \,\,\,\,\,\,\,\,\,\textbf{for} partition $j$ in accessed
partitions \textbf{do}
\State \,\,\,\,\,\,\,\,\,\,\,\,\,\,\textbf{if} $i==j$ \textbf{then}
\textsf{skip}
\State \,\,\,\,\,\,\,\,\,\,\,\,\,\,\textbf{if} $V_{b_i}^{X_i}[X_j] >
V_{b_j}^{X_j}.LCE $ \textbf{then}
\State
\,\,\,\,\,\,\,\,\,\,\,\,\,\,\,\,\,\,\,Unsatisfied\_dependencies
$\leftarrow$ $<X_j, V_{b_i}^{X_i}[X_j]>$
\State \,\,\,\,\,\textbf{if} Unsatisfied\_dependencies is not empty
\textbf{then}
\State
\,\,\,\,\,\,\,\,\,\,\textsf{rot\_second\_round(}unsatisfied\_dependencies)
\State \}
\end{algorithmic}
\end{small}
\end{algorithm2e}

\subsubsection{\textbf{Read-Only Transaction Protocol}}
\cut{
In this section, we present the protocol to process read-only
transactions.

\textbf{Client Protocol.}}

The read-only transaction protocol discussed in this section is designed to guarantee serializability. The two rounds needed to perform a read-only transaction (in the worst case) ensure that the objects retrieved in the two rounds are always serializable. When a client issues a distributed read-only transaction, it sends a
request to the leader of each accessed partition. The leader of each
partition responds with the current values, and the most recent Merkle
Tree information. The leader also sends the CD
vector that corresponds to the returned Merkle Tree root. The client
uses the returned CD vectors to decide whether the returned values
are consistent across partitions as we show next.

\textbf{Verifying Dependencies.}
The algorithm to verify dependencies (Algorithm~\ref{alg:verify})
processes dependency vectors one by one. The client receives $n$
dependency vectors ($V$ in the algorithm) from $n$ partition leaders,
where $n$ is the number of accessed partitions in the read-only
transaction. (We use the notation $\mathcal{V}_{b_i}^{X_i}$ to
denote the dependency vector of the $i^{th}$ partition accessed by
the read-only transaction.)
For each dependency vector, the algorithm verifies dependencies to
other accessed partitions. Specifically, when processing
$\mathcal{V}_{b_i}^{X_i}$, the algorithm checks the dependencies to
all other $n-1$ partitions---hence, it checks
$\mathcal{V}_{b_i}^{X_i}[X_j]$, for all $0 \le j < n$ and $j \ne i$.

Each value $\mathcal{V}_{b_i}^{X_i}[X_j]$ is compared with the
\emph{Last Committed Epoch} (LCE) received from $X_j$. The LCE is the
batch number that corresponds to the batch where the committed
records in the received batch have prepared.  
If the LCE value is
greater than or equal to $\mathcal{V}_{b_i}^{X_i}[X_j]$, then the
dependency is satisfied, and the algorithm proceeds to check the next
dependency.  If not, then the dependency is flagged as unsatisfied
and becomes part of the second-round of the read-only transaction
algorithm. 

\textbf{Termination and the Second Round.}
After all the dependency vectors are checked, the algorithm
terminates if all dependencies are satisfied. Otherwise, it starts
the second round of the read-only transaction algorithm. In the
second round, the algorithm asks explicitly for the missing
dependencies. For example, if the dependency
$\mathcal{V}_{b_i}^{X_i}[X_j]$ was not satisfied in the first round,
the algorithm sends a request to the leader of partition $X_j$ asking
for batch number $\mathcal{V}_{b_i}^{X_i}[X_j]$. After all such
batches are served, the algorithm terminates.

\subsection{Properties of Read-Only Transactions}
\label{sub:protocol_properties}

In this section, we discuss the correctness and data freshness guarantees of TransEdge.

\textbf{Read-only Transaction Correctness.}
TransEdge guarantees serializability~\cite{bernstein1987concurrency}. We extend the proof presented in Section~\ref{sub:rw-correctness}. 
%
The intuition behind the proof is to show that read-only transactions do not introduce cycles in the serializability graph (SG) of any transaction execution history (history for short). 


\begin{lemma}
\label{lem:cdvector}
For any two read-write transactions with a conflict $t_i \rightarrow t_j$, the CD vector of the batch that includes $t_j$, $b^Y_j$, includes the dependency to the batch that includes $t_i$, $b^X_i$ as well as the dependencies of $b^X_i$.
\end{lemma}
\begin{proof}
This follows from the design where the CD vector is updated after preparing with the dependencies of all the shards that are involved in the transaction.
\end{proof}

\begin{lemma}
\label{lem:cdvector2}
For a sequence of read-write transaction conflicts $t_i \rightarrow \ldots \rightarrow t_j$, the CD vector of the batch that includes $t_j$, $b^Y_j$, includes the dependency to the batch that includes $t_i$, $b^X_i$.
\end{lemma}
\begin{proof}
This is the case by applying Lemma~\ref{lem:cdvector} transitively. 
\end{proof}

\begin{lemma}
Given a serializability graph of a transaction history of TransEdge read-write transactions, adding the node and edges related to a read-only transaction $t_r$ would not introduce any cycles.
\end{lemma}
\begin{proof}
We start with the serializability graph $SG^{rw}$ which contains the read-write transactions only of a TransEdge execution. Given our discussion in Section~\ref{sub:rw-correctness}, $SG^{rw}$ contains no cycles. When adding $t_r$ to $SG^{rw}$, we get the graph $SG^{rw}_{+t_r}$. In $SG^{rw}_{+t_r}$, a new node for $t_r$ is added as well as conflict relations to other transactions. The following are the edges that are added for conflicts between $t_r$ and other transactions in $SG^{rw}_{+t_r}$:
(1)~Write-read conflicts: these are conflicts that are added from a set of transactions $T_{wr}$ that wrote objects that were read by $t_r$. (2)~Read-write conflicts: these are conflicts that are added from $t_r$ to a set of transactions $T_{rw}$ that overwrote the values read by $t_r$. 

We need to show that introducing these edges would not introduce a cycle to $SG^{rw}_{+t_r}$. We prove this by contradiction. Assume to the contrary that introducing the edges above leads to a cycle. This cycle would look like the following $t_r \rightarrow_{rw} t_{after-r} \rightarrow \ldots \rightarrow t_{before-r} \rightarrow_{wr} t_r$. Note that such an introduced cycle must have this pattern where $t_r$ is part of the new cycle (since cycles did not exist before), and where the edge emanating from $t_r$ is a read-write edge and the edge coming into $t_r$ is a write-read edge (this is because all the operations in $t_r$ are read-only operations.) 

From Lemma~\ref{lem:cdvector2}, since there is a sequence of conflicts from $t_{after-r}$ to $t_{before-r}$ this means that the CD vector of the batch containing $t_{before-r}$, $b^X_i$, contains a dependency to the batch containing $t_{after-r}$, $b^Y_j$. There are two cases:\\
(1)~the read from batch $b^X_i$ (for $t_{before-r}$) was in the first round of the read-only transaction: in that case the client would have asked for batch $b^Y_j$ since it is a dependency and the transaction would have read from $t_{after-r}$. This is a contradiction since as part of our starting assumption the conflict is from $t_r$ to $t_{after-r}$. \\
(2)~the read from batch $b^X_i$ (for $t_{before-r}$) was in the second round of the read-only transaction: this means that there is a read in $t_r$ from a transaction $t_k$ that triggered the second round read from $t_{before-r}$ (i.e., the CD vector of the batch containing $t_k$, $b^K_k$, has a dependency to $b^X_i$.) However, since TransEdge Algorithms includes the dependency to the batch in the conflict (e.g., from $b^X_i$) as well as that batch's dependencies (e.g., $b^X_i$'s dependency to $b^Y_j$), this means that $b^X_i$ must have been part of the dependencies in the CD vector of $b^K_k$. In this case, the second round of read-only transactions should have also asked for $b^K_k$. This is a contradiction since $t_k$ that is part of $b^K_k$ is not part of the read response since in our starting assumption the conflict is from $t_r$ to $t_{after-r}$.\\
Therefore, we show that both cases lead to a contradiction proving that no cycles are introduced.
\end{proof}

\begin{theorem}
TransEdge with both read-write and read-only transactions guarantee serializability.
\end{theorem}
\begin{proof}
Starting from a serializability graph (SG) of read-write transactions only, we know that $SG$ is serializable from Section~\ref{sub:rw-correctness}. Then, by Lemma~\ref{lem:cdvector2} we know that adding a read-only transaction to $SG$ would not lead to a conflict. Repeating this for every read-only transaction gets us to the serializability graph with all read-write and read-only transactions that includes no cycles. 
\end{proof}

\subsubsection{\textbf{Guarantee of two-round reads}.} TransEdge needs at most two rounds to produce a consistent read-only transaction response. We now prove that the responses received in the second round would not lead to further unsatisfied dependencies and therefore, there is never a need for a third round.

\begin{theorem}
In TransEdge, if there is a second round in the read-only transaction, then there will be no further dependencies and a third round is never needed.
\end{theorem}
\begin{proof}
Assume to the contrary that after receiving a batch $b^X_i$ in the second round that there is an unsatisfied dependency in $b^X_i$ that was not received in either the first or second rounds. This means that the CD vector of $b^X_i$ contains a dependency to a batch $b^Y_j$ that is not part of the response in the first or second rounds. However, consider batch $b^Z_k$ which is the batch in the first round that had the missing dependency leading to needing to ask for $b^X_i$ in the second round. Based on the TransEdge protocol, the dependencies of $b^Z_k$ include all the dependencies of its dependents (since the pairwise max was taken when calculating the CD vector of $b^Z_k$). Therefore, $b^Y_j$ is part of the dependencies of $b^Z_k$ and if it was not part of the first round responses, then the client would have asked for it in the second round. This is a contradiction since there is no need to ask for $b^Y_j$ in a third round.
\end{proof}

\subsubsection{\textbf{Freshness in TransEdge}}
The read-only transaction algorithm in TransEdge can guarantee consistency across a database snapshot. However, it cannot guarantee that the response from participating replicas always includes the latest updates. 

\begin{sloppypar}
Malicious participating replicas might return an old---albeit consistent---snapshot. A way to
enforce a freshness guarantee is to
add a timestamp to each batch that represents the time this
snapshot committed. Although, there are many other ways to solve this problem we consider this orthogonal to the scope of this project. The leader includes a timestamp of the current
time in the batch when it is sent to the other BFT replicas. The BFT
replicas verify that the timestamp in the batch is within a window of
time compared to their clock. (For example, a BFT replica $r$ accepts a
batch only if the timestamp of the batch is within 30 seconds of $r$'s
clock.). This ensures that a malicious leader is restricted to a
specific window when choosing the timestamp. Using these timestamps
and the configured time window, the read-only clients can establish a
guarantee on the freshness of the data. Note that such a guarantee
would not ensure that the returned batch is the most recent one, but it
ensures that the batch was committed within a recent time window.
\end{sloppypar}

\cut{This is because it is not possible that the batches
from the second round would result in unsatisfied dependencies that
conflict with the read-only transaction. 
}

\cut{
There are two possibilities at the end of the first round: either the read-only transaction successfully terminated or a batch number was found for an inconsistent read-only transaction key. In the latter case, the second round allows us to request for the key with the exact batch number when it was committed. Thus, a read-only transaction would terminate in the second round and does not need additional rounds to complete. 
}

\cut{
\subsection{\textbf{Reducing the Dependency Vector Size}}
\label{sub:opt}

A large number of partitions, makes it possible for the size of
dependency vectors to grow to cause a communication overhead. We can reduce the 
size of the dependency vector in several ways.
The first optimization is to send only the relevant entries when a dependency
vector is sent from a leader to a client. 
The relevant entries are the ones corresponding to the partitions accessed by the client's
read-only transaction. The client algorithm does not need the other
entries.

The other optimization is to reduce the size of the dependency vector
in the edge nodes. Assume that the number of edge nodes is $n$. In
the normal-case---what we presented so far---each edge node will
maintain a dependency vector with a size equal to $n$.
However, in many workloads, distributed transactions exhibit some
locality features where groups of partitions are typically accessed
together. For example, for some edge node $i$, it might be the case
that most distributed transactions that access $i$ only access $m$
other edge nodes, where $m$ is much smaller than $n$ ($m<<n$). In
this case, we can maintain only the dependencies of these $m$
partitions in $i$, rather than maintaining all $n$ dependencies.

A consequence is that
a distributed read-only transaction that accesses edge node $i$ and
another partition outside of its $m$ partitions cannot be executed
with the snapshot read-only algorithm. Rather, such read-only
transactions will have to execute as a regular read-write transaction
using 2PC (Section~\ref{sec:txn}).
}



\cut{
\begin{table}[!t]
\begin{center}
        \begin{tabular}{ c | c  c  c  c  c  c  c  }
        &        $C$ & $O$ & $V$ & $T$ & $I$ & $S$ & $M$ \\ \hline
        $C$ & 0   & 19  & 62  & 113 & 134 & 183 & 249 \\
        $O$ & 19  & 0   & 117 & 104 & 133 & 161 & 221 \\
        $V$ & 62  & 117 & 0   & 172 & 71  & 244 & 182 \\
        $T$ & 113 & 104 & 172 & 0   & 214 & 67  & 124 \\ 
        $I$ & 134 & 133 & 71  & 214 & 0   & 179 & 120 \\ 
        $S$ & 183 & 161 & 244 & 67  & 179 & 0   & 54  \\ 
        $M$ & 249 & 221 & 182 & 124 & 120 & 54  & 0   \\ 
        \end{tabular}
        \caption{The average Round-Trip Times in milliseconds for
every pair of the 7 datacenters (zones).}
        \label{tab:rtts}
\end{center}
\end{table}
}

\section{Experimental Evaluation}\label{sec:eval}

\cut{In this section, we present a performance evaluation of a prototype of TransEdge.We report latency and throughput results while running workloads for local and distributed read-write transactions as well as distributed read-only transactions.} TransEdge uses BFT-SMaRt~\cite{bft-smart-lib} as the BFT system to commit batches to the state-machine replication log within clusters. Therefore, we inherit the fault-tolerance and state-machine replication processes of BFT-SMaRt. We also compare the performance of TransEdge's read-only transactions with a coordination-based read-only transaction protocol. We term this system 2PC/BFT, and it aims to mimic how existing hierarchical BFT systems perform read-only operations~\cite{nawab2019blockplane,gupta13resilientdb,
al2017chainspace, amir2010steward,amiri2019parblockchain} (see Section~\ref{sub:compare}.) The 2PC/BFT system has the same structure as TransEdge, however, the system performs read-only transactions by coordinating with other leaders in other partitions via two-phase commit. This allows us to contrast the performance of read-only transactions in TransEdge with those executed in a coordination based system. We also compare Read-only transactions in TransEdge with Augustus\cite{padilhaaugustus2013}, a BFT system that supports fast read-only transactions and has similar data partitioning system as TransEdge.
\cut{
The prototype implementation of TransEdge does not implement leader election protocol for the cluster partition, however, the BFT-SMaRt\cite{bft-smart-lib} used to process PBFT performs leader election when processing messages within the partition.
The idea of a malicious 2PC leader is not in the scope of this implementation. 
}
\subsection{Experimental setup}

\textbf{Setup.}
To evaluate transactions in TransEdge we use 5 clusters with 7 replicas in each cluster. The 7 replicas in a cluster allows the cluster to support Byzantine faults of up to 2. Experiments are performed on ChameleonCloud\cite{keahey2020lessons} using Cascade Lake R machines with Xeon Gold 6240R processor, 192GB RAM, 96 Threads. Transaction workload is generated by 2 clients running 10 threads for read-only transactions and read-write transactions.

\textbf{Data and Transaction Model.}
There are two main types of transactions that are evaluated in TransEdge: Read-write transactions and read-only transactions.
Read-write transactions are of the following types: local write-only transactions, local read-write transactions, and distributed read-write transactions. 
Write-only and local read-write transactions are transactions that operate on keys local to a cluster. Distributed read-write transactions need to be executed in coordination with other clusters as they contain operations performed on keys from other clusters.
Read-only transactions read \textit{n} unique keys from \textit{m} clusters. 

Each read-write transaction contains 5 read and 3 write operations distributed across 5 clusters. Each read-only transaction contains 5 read operations reading 1 key from each cluster.

\cut{
\textbf{Initialization.}
Write-only transactions are used to initialize the database before beginning the experiments on local and distributed read-write transactions. This allows the local and distributed read-write transactions to operate on keys already existing in the database.
}

\textbf{Workload.}
\cut{
To stress the system during experiments, the distributed read-write transactions workload is designed to be a coordination intensive workload. A batch of distributed read-write transactions has on average 85\% transactions which are distributed (requiring two-phase commit for execution) and 15\% transactions that are local to a cluster. Each transaction in this workload contains 5 read and 3 write operations. The transaction batches processed by TransEdge are uniformly distributed across all clusters so as to not skew batch processing.}

The workload to test the system contains $500k$ transactions. Total number of keys in the clusters is $1M$. Keys are uniformly distributed across the clusters using hashing. Key and values used in the transactions have a size of 4 bytes and 256 bytes respectively. The workload generator is inspired by YCSB~\cite{ycsb2010} and its transactional extensions~\cite{das2011albatross}. The workload generator generates operations based on the provided ratios. A key for each operation is also picked randomly. Then, a group of operations are bundled into a transaction~\cite{das2011albatross}.

\cut{
Figures \ref{fig:throughput_wo_lrwt} and \ref{fig:latency_wo_lrwt} show results where write-only transactions are executed before the local read-write workloads are executed. Figures~\ref{fig:throughput_drwt} and~\ref{fig:latency_drwt} describe results where distributed transaction workload is executed after the execution of the local write-only transactions. Results for the read-only transaction described in figure \ref{fig:latency_rot} are obtained by first running the write-only transaction workload and subsequently executing the read-only transaction workload.
}


 
\subsection{Experimental Results}
\cut{We start by showing results for read-write transactions to establish a baseline of TransEdge's performance and then evaluate the performance of read-only transactions which is the main contribution of this paper.}


\textbf{Read-only transactions.} 
We first present the results of experiments for read-only transactions---the main focus of this paper. The results measure the end-to-end latency of read-only transactions from the client's side. 

The results of the first set of experiments is shown in Figure~\ref{fig:rot_2pc}. In it, we measure the read-only transaction latency of TransEdge and demonstrate how it outperforms read-only transactions that are performed as 2PC/BFT transactions. 
TransEdge outperforms 2PC/BFT by up to 24x when accessing two clusters and outperforms 2PC/BFT by 9x when accessing five clusters. 

The reason for this performance difference is that running read-only transactions as regular transactions incurs overheads due to BFT agreement and 2PC coordination that are similar to read-write transactions. Read-only transactions executed as regular transactions require coordination among the nodes involved in the transaction. The cost of this coordination leads to a significant increase in latency of read-only transactions.

Figure~\ref{fig:rot_2pc} also shows that the latency increases as the number of clusters accessed by the read-only transactions increases. This is expected as the number of replicas that need to be involved in the 2PC process would involve more messages exchanged by the replicas. This is reflected in the graphs for latency of both TransEdge and 2PC/BFT read-only transactions. The average latency of the snapshot read-only transactions over 2PC/BFT is between 69-82ms when accessing more than one cluster. The latency is quite large when compared to the snapshot read-only transactions executed by TransEdge. The main reason for the large overhead in latency is due to the amount of coordination time that is needed to coordinate regular transactions as opposed to efficient snapshot read-only transactions in TransEdge.
\cut{
\begin{table}[hb]
\begin{center}

  
  \begin{tabular}{|l|c|c|c|c|}
  \hline
  Number of clusters & 2 & 3 & 4 & 5 \\ \hline
  Latency Speedup & 20.7x & 14.7x & 10.2x & 8.65x \\ \hline
\end{tabular}

  \caption{Comparison of speedup in latency of Read-only Transactions as executed by TransEdge over 2PC/BFT}
  \label{tab:rot_speedup}
\end{center}

\end{table}

Table~\ref{tab:rot_speedup} summarizes the comparison of speedups between snapshot read-only transactions executed by TransEdge over the 2PC/BFT system. As seen from the table the speedup in latency depends on the number of clusters accessed. The highest is achieved when accessing only 2 clusters and the lowest when accessing the maximum (in our experiments) of 5 clusters. It can be seen that even when accessing 5 clusters the speedup in latency of TransEdge over 2PC/BFT is $8.6$x. 
}

\begin{figure}[!t]
\begin{center}
\includegraphics[scale=0.45]{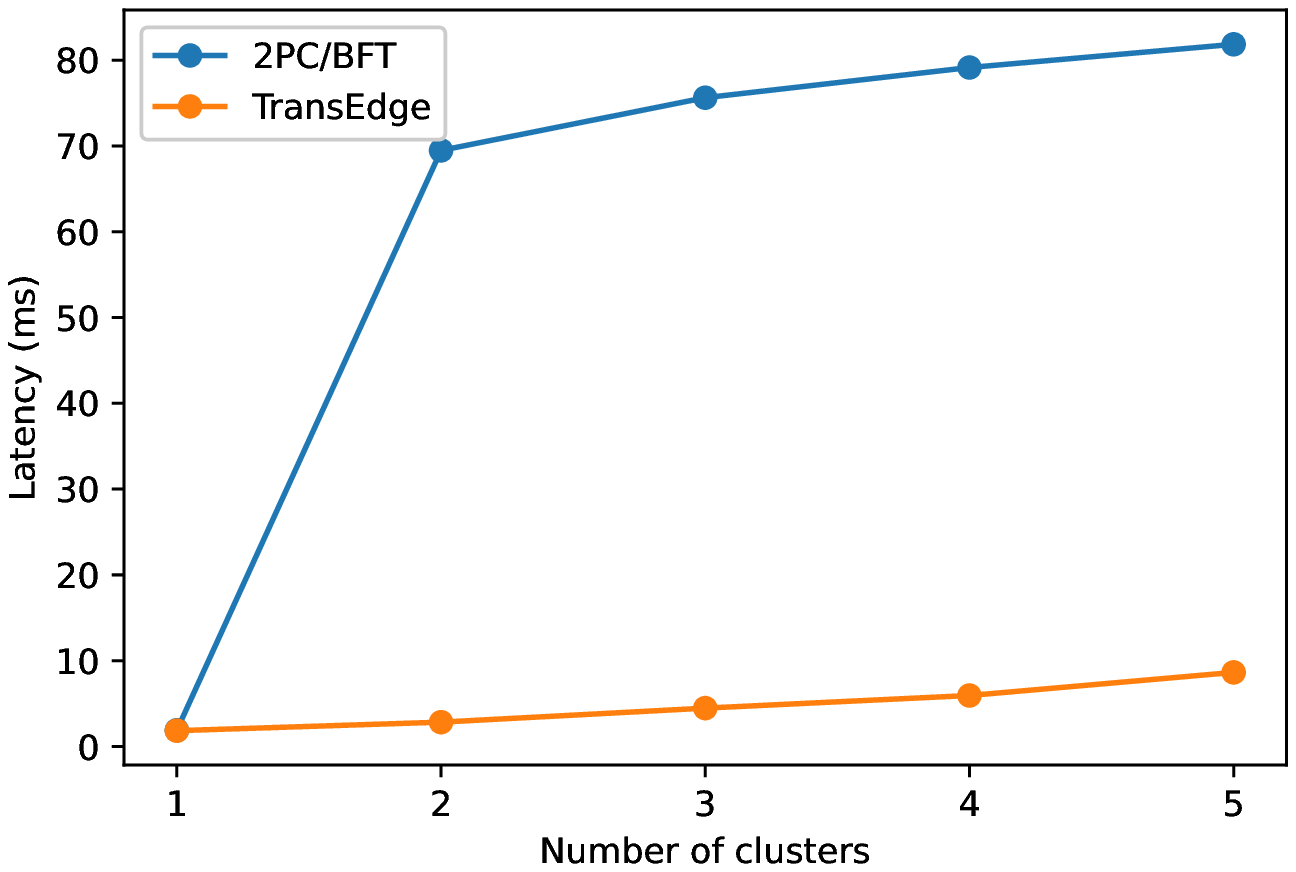}
\caption{Comparison of average latency of read-only transactions executed over a 2PC/BFT system and TransEdge.} 
\label{fig:rot_2pc}
\end{center}
\end{figure}

\textbf{Latency of read-only transaction rounds.}
Since read-only transactions have a maximum of two rounds, we perform experiments to measure how much each round contributes to the average latency of read-only transactions. In Figure~\ref{fig:latency_rot}, we show the latency of the first round of the read-only transactions as the light blue bar and show the additional latency of verifying and correcting the inconsistencies in the second round of execution as the orange bar. Not all read-only transactions in TransEdge require round 2 and the orange bar in the figure computes the effective latency of round-2 communication in TransEdge. This is computed by multiplying the percentage of read-only transactions that require round-2 with the additional latency of round-2 execution. 

We also compare the latency of read-only transactions in TransEdge with Augustus\cite{padilhaaugustus2013}. The latency of read-only transactions in Augustus are shown in the dark blue bar in figure~\ref{fig:latency_rot}. Figure~\ref{fig:rot-t-vs-augustus} shows the corresponding comparison of throughput in TransEdge and Augustus. The improvement of read-only transactions in TransEdge over Augustus can be attributed to a lock-free and coordination-free execution of read-only transactions. This is seen in the higher throughput and lower latency in TransEdge when compared to Augustus when accessing a single cluster. A single cluster Read-only transaction refers to single partition data access. We see in both figures~\ref{fig:latency_rot} and \ref{fig:rot-t-vs-augustus} that Augustus performs poorly even on single partition reads and the performance degrades with multipartition reads. Read-only transactions in TransEdge also do not interfere with read-write transaction execution and thus are not affected (or affect) read-write transaction execution. In order to test this hypothesis, we executed a series of long-running read-only transactions with a large set of keys in conjunction with read-write transactions. The results of this experiment are shown in figure~\ref{fig:scanning-rot-t-vs-augustus}. Figure \ref{fig:scanning-rot-t-vs-augustus} shows the effective latency of TransEdge read-only transactions which includes both round 1 and round 2 of the read-only transactions. TransEdge's latency is a function of dependency computation as opposed to Augustus which uses shared locks as a mechanism to coordinate between different clusters. The advantage of TransEdge over Augustus is that TransEdge does not use locking thus ensuring that read-only transactions do not conflict with read-write transactions. This is shown in the percentage of Aborted transactions during the execution of long-running read-only transactions.  The percentage of aborts caused due to conflicting read-only transactions in TransEdge is 0 and a comparison with Augustus is shown in Table~\ref{tab:aborts-t-vs-augustus}. This does not mean that TransEdge read-write transactions do not abort as is shown in figure~\ref{fig:transedge-aborts-drwt}. These experiments show that main benefit of TransEdge comes from coordination-free execution of read-only transactions and this is visible in the increased throughput, lower latency and lower percentage of aborts. Figure~\ref{fig:rot-throughput-latency} shows the impact of network latency on throughput of read-only transactions in TransEdge. The main overhead of read-only transactions in TransEdge is the computation of consistency in read-only transaction keys using the dependency vector thus the drop in throughput is not as high as that of the read-write transactions seen in Figure~\ref{fig:throughput_drwt_latency}.
\begin{table}[hb]
\begin{center}
  \small
  \begin{tabular}{|l|c|c|c|c|c|}
  
  \hline
  Number of clusters & 1 & 2 & 3 & 4 & 5 \\ \hline
  Augustus & 0.8 & 1.3 &  2.15 & 3.4 & 4.27 \\ \hline
  TransEdge & 0 & 0 & 0 & 0 & 0 \\ \hline
\end{tabular}
\vspace{8mm}
  \caption{Comparison of aborts in read-write transactions in TransEdge and Augustus caused due to conflicting read-only transactions.}
  \label{tab:aborts-t-vs-augustus}
\end{center}

\end{table}

  




\cut{The graph in Figure \ref{fig:latency_rot} shows the case where the consistency check and query for the consistent version of the keys are performed for 75\% of read-only transactions. In order to emulate inconsistencies in the system, the distributed read-write transactions are executed in conjunction with the read-only transaction workload. The execution of the distributed read-write transactions created the inconsistencies in the database which triggers the second round queries for the read-only transactions. 

Figure \ref{fig:latency_rot_consistent} on the other hand shows the latency of read-only transaction execution when executed on a system where no new writes were written. Figure \ref{fig:latency_rot_consistent} aims to show the overhead of performing the second round of the read-only transactions without the additional query to retrieve consistent data from the TransEdge clusters.}

\cut{
The throughput of ready-only transactions decreases with the increase in the number of clusters accessed as shown in Figure \ref{fig:throughput_rot}. This is expected as the overhead of coordination and verifying the consistency of reads across clusters increases.
In particular, the throughput that a single client can achieve drops from 260 read-only transactions per second when accessing one cluster to 36 read-only transactions per second when accessing five clusters.}

\cut{The benefit of our design of read-only transactions becomes clear by observing their latency. Read-only transactions can get a consistent response within 28ms even while read-only transactions span 5 different clusters. Compared to read-write transactions where latency is beyond 160ms, read-only transactions are much more efficient. The reason for this efficiency is the design of being commit-free and non-interfering. }


\begin{figure}[!t]
\begin{center}
\includegraphics[scale=0.45]{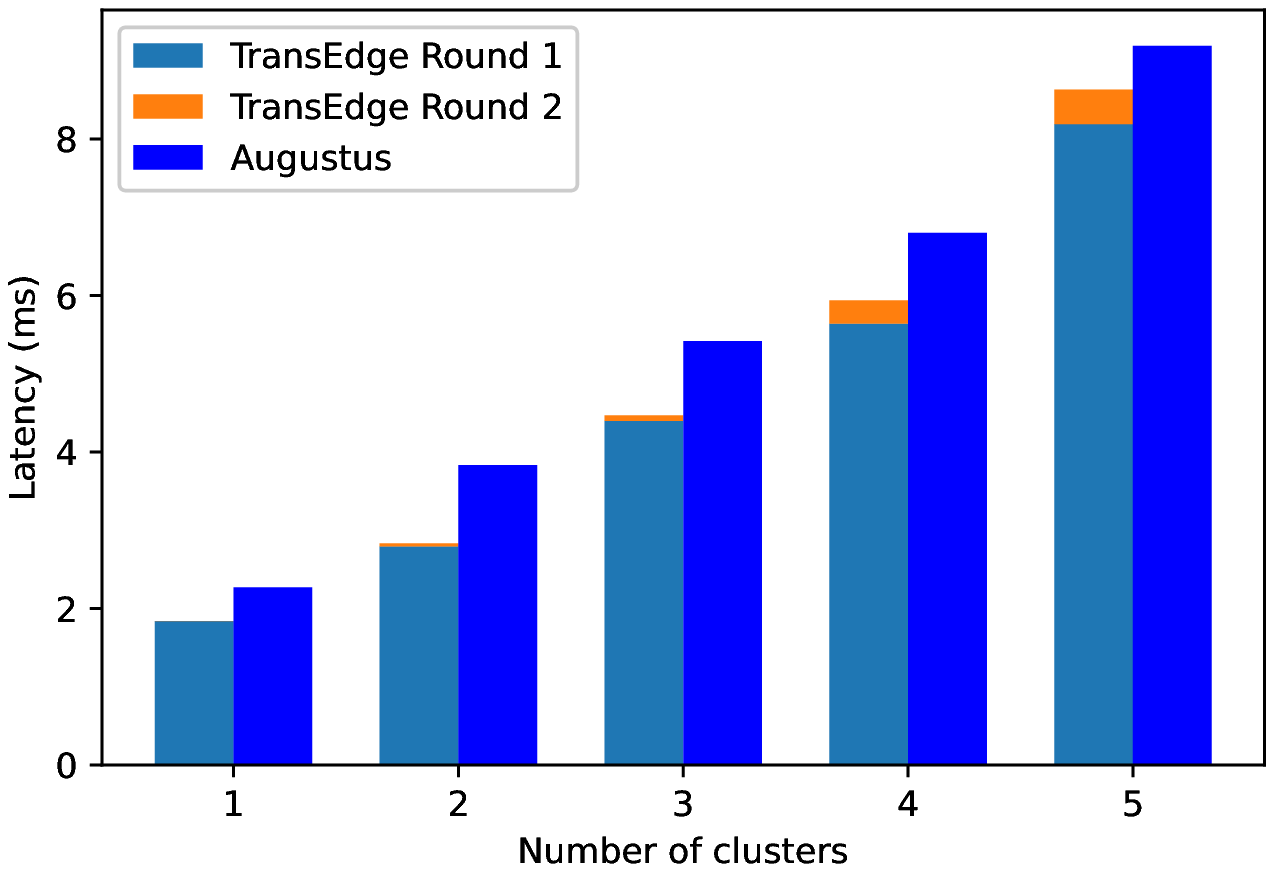}
\caption{Comparison of average latency (in milliseconds) of read-only transactions in TransEdge and Augustus\cite{padilhaaugustus2013} when varying the number of clusters accessed.} 
\label{fig:latency_rot}
\end{center}
\end{figure}


\begin{figure}[!t]
\begin{center}
\includegraphics[scale=0.45]{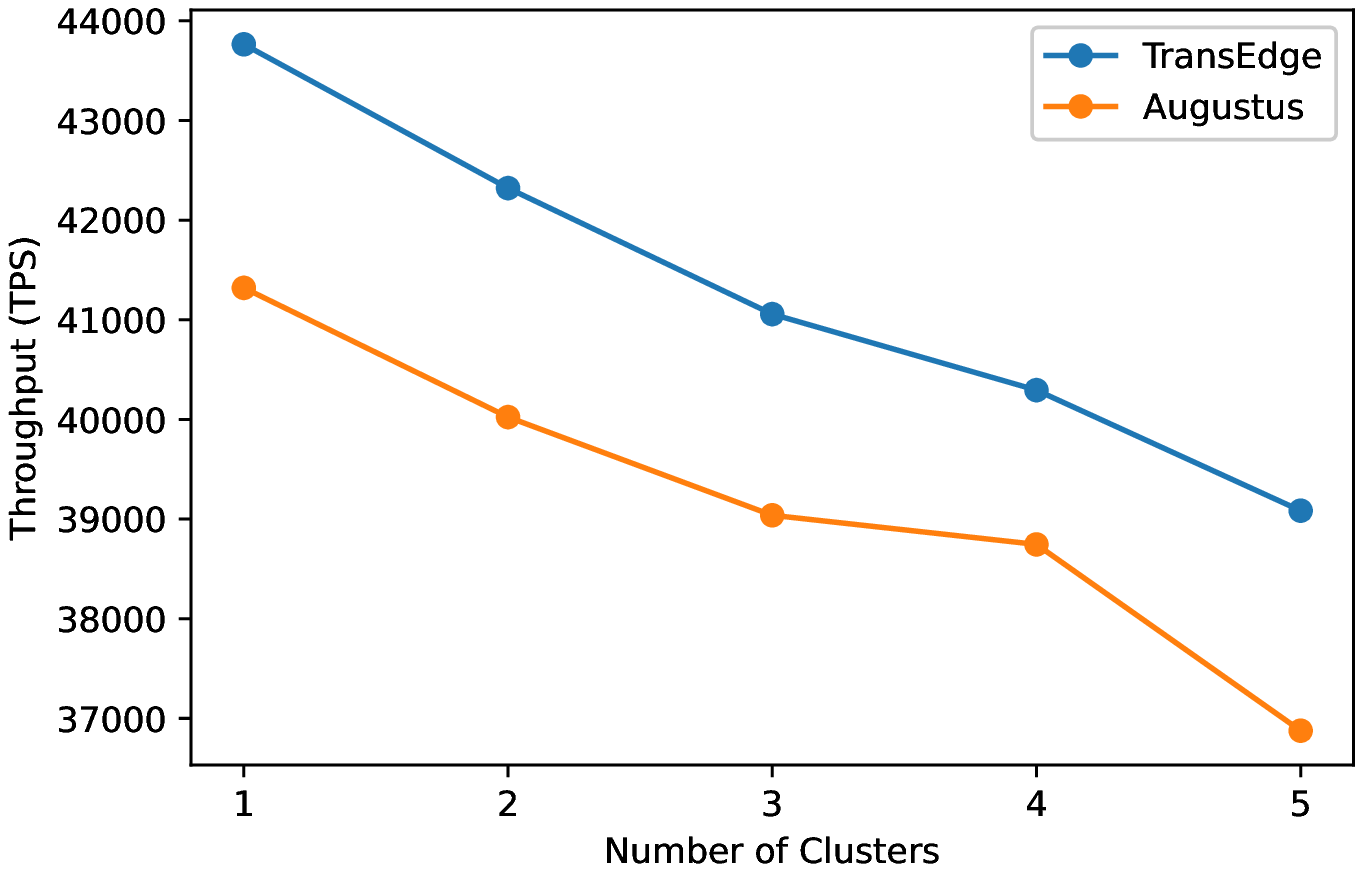}
\caption{Comparison of average throughput (in transactions per second) of read-only transactions in TransEdge and Augustus\cite{padilhaaugustus2013}. The horizontal axis describes the number of clusters accessed by the read-only transactions.} 
\label{fig:rot-t-vs-augustus}
\end{center}
\end{figure}

\begin{figure}[!t]
\begin{center}
\includegraphics[scale=0.45]{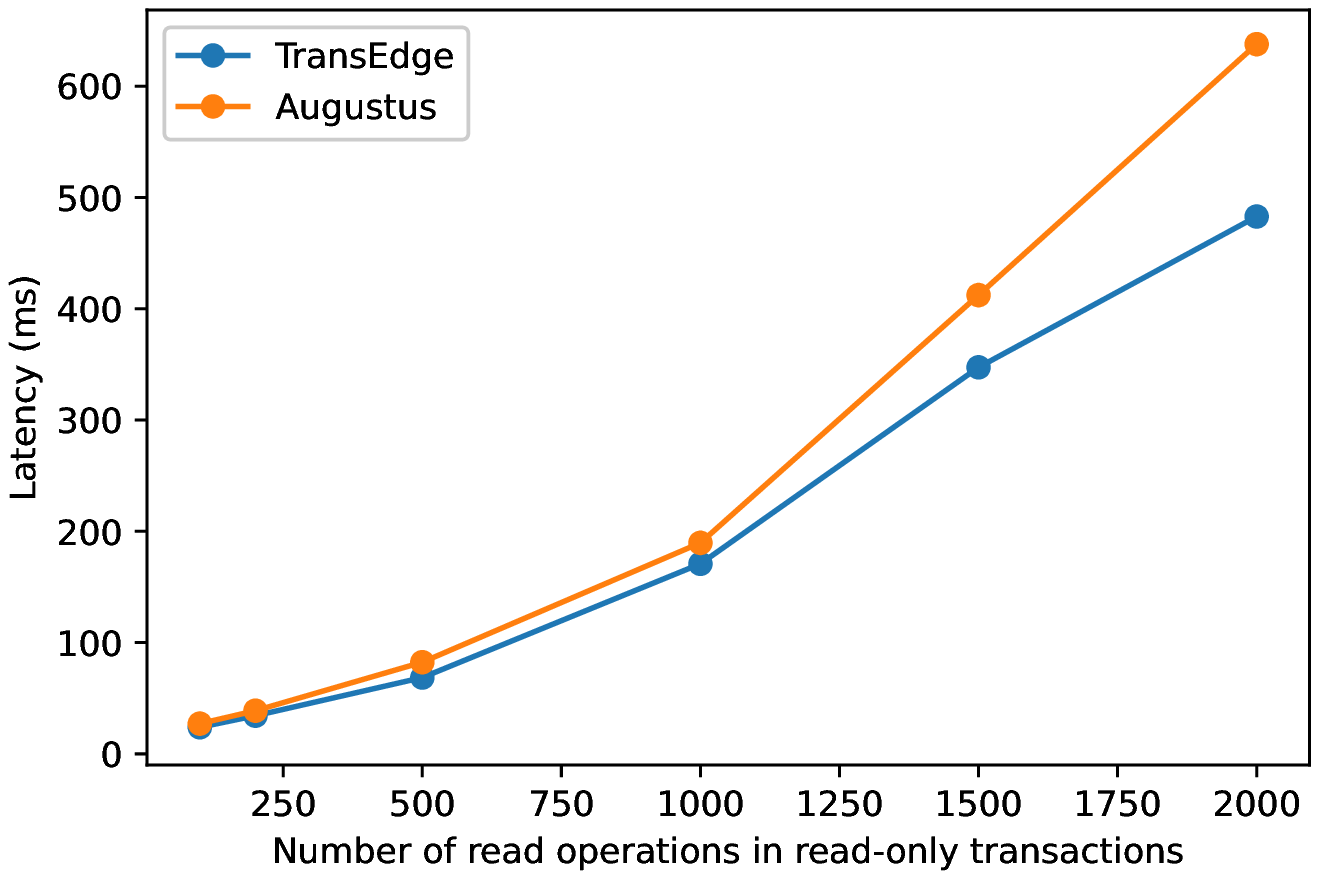}
\caption{Comparison of average latency (in milliseconds) of long-running read-only transactions in TransEdge and Augustus\cite{padilhaaugustus2013}.} 
\label{fig:scanning-rot-t-vs-augustus}
\end{center}
\end{figure}
\vspace{5mm}
\begin{figure}[!t]
\begin{center}
\includegraphics[scale=0.45]{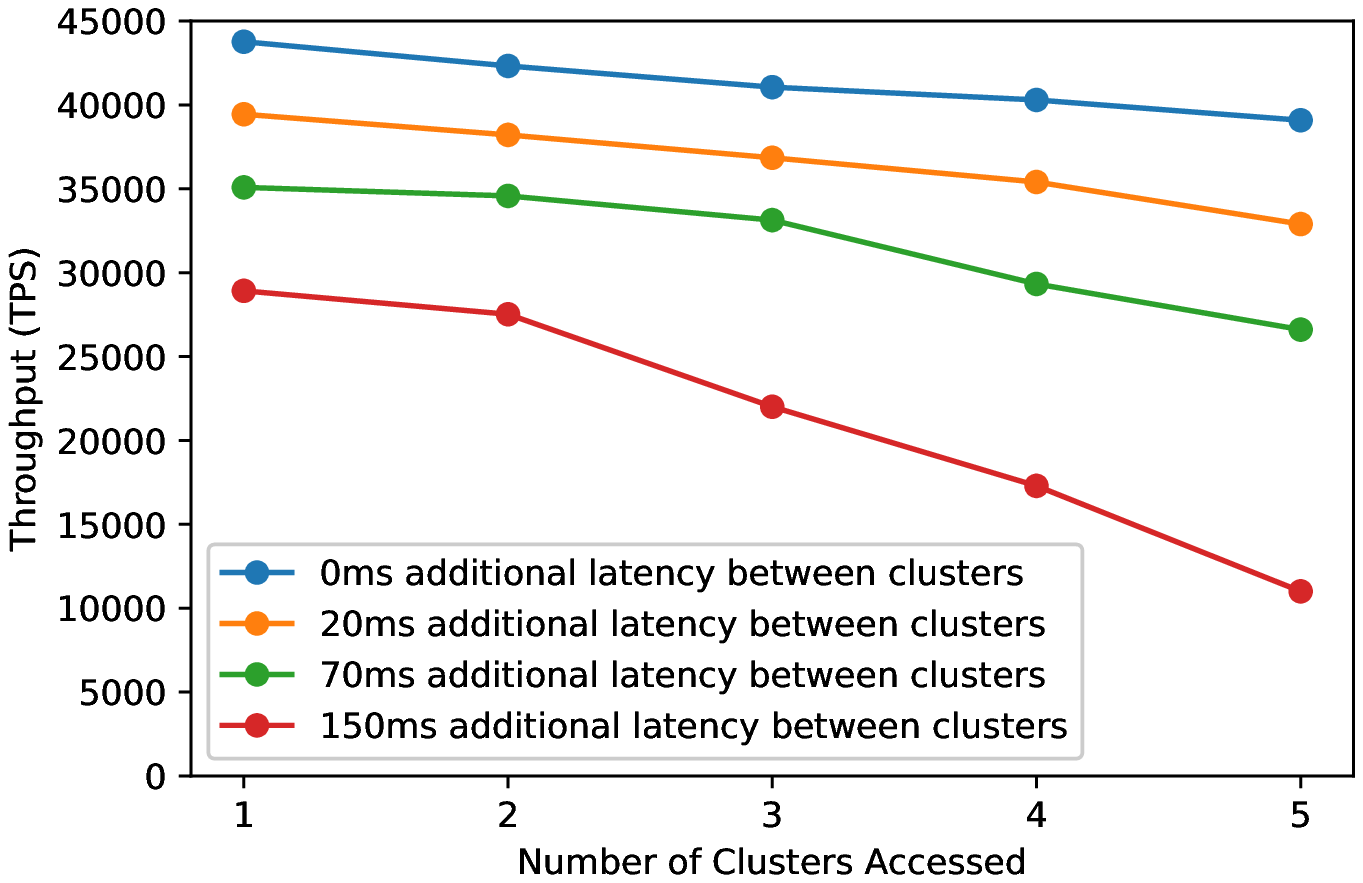}
\caption{Variation in throughput of read-only transactions as the latency between clusters is increased.} 
\label{fig:rot-throughput-latency}
\end{center}
\end{figure}



\textbf{Write-only and local read-write transactions.}
%
Figure~\ref{fig:throughput_wo_lrwt} show the results of write-only and local read-write transaction workloads. TransEdge processes data in batches and the horizontal axis lists the batch size of transactions executed by the transaction processor. In figure \ref{fig:throughput_wo_lrwt}, we observe that both types of transactions perform similarly. This is because both transaction types are local and undergo the same commitment pattern via the BFT process. Specifically, write-only and local read-write transactions reach their peak throughput of around 2000--2500 transactions per batch. However, write-only transactions perform slightly better than local read-write transactions as the number of batches grow. This is because read-write transactions require more coordination to guarantee serializability in the presence of read operations. Figure~\ref{fig:throughput_wo_lrwt} also shows the average throughput of local read-write transactions in the 2PC/BFT system. The 2PC/BFT system performs similarly to TransEdge as they follow similar steps for commitment. \cut{In both figures as the number of transactions in the batch is increased, the average latency of the transaction execution increases. This is because as we increase the size of the batch, the overhead required to process, coordinate, and apply transactions increases.}

\begin{figure}[!t]
\begin{center}
\includegraphics[scale=0.45]{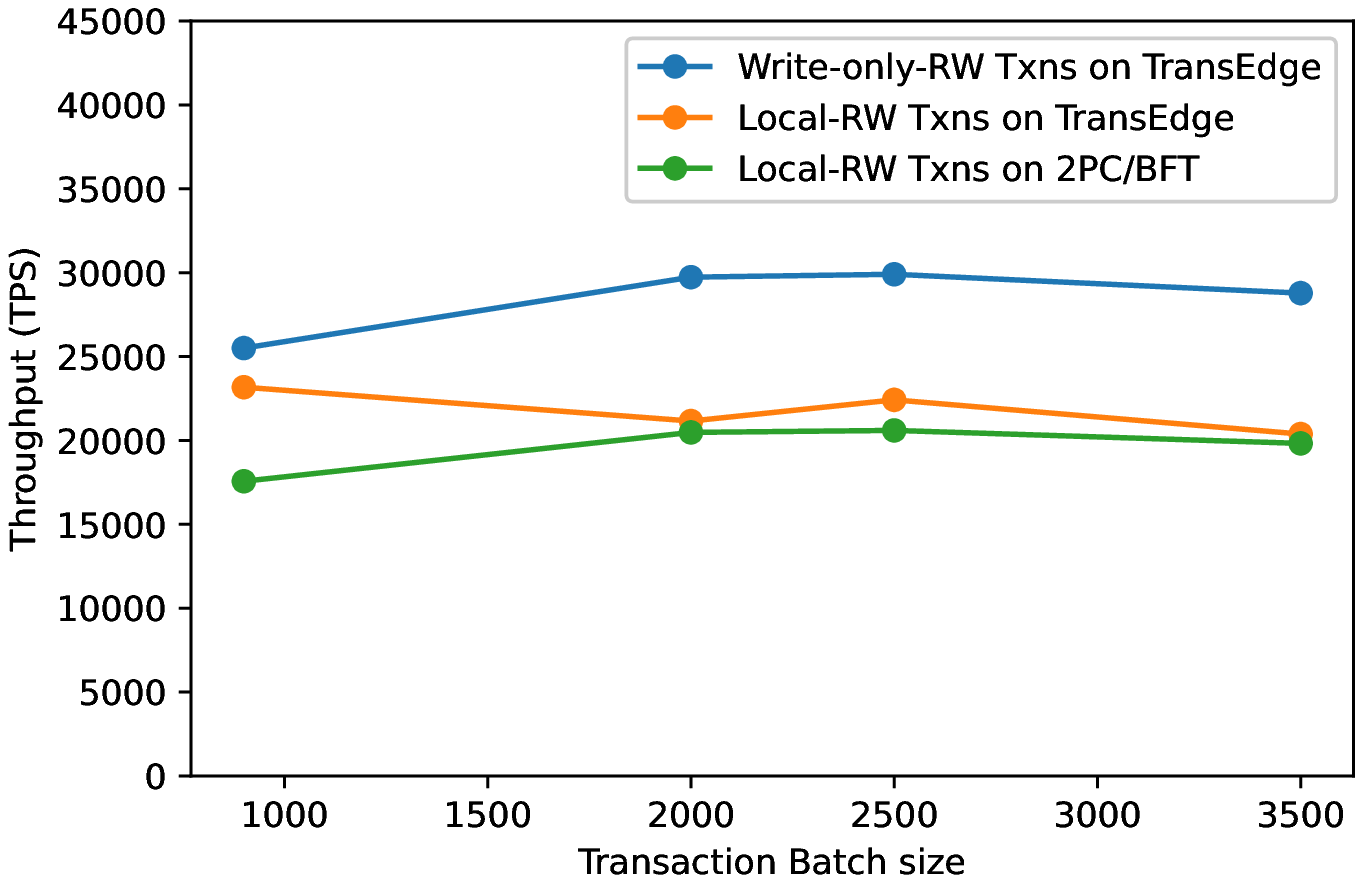}
\caption{Average throughput of write-only, local read-write transactions on TransEdge and 2PC/BFT system.}
\label{fig:throughput_wo_lrwt}
\end{center}
\end{figure}


\textbf{Distributed read-write transactions.}
We performed multiple experiments with distributed read-write transactions. Figures~\ref{fig:latency_drwt_skew} and~\ref{fig:throughput_drwt_skew} are from experiments where we vary the number of read and write operations within the read-write transactions. Figure~\ref{fig:transedge-aborts-drwt} and~\ref{fig:throughput_drwt_latency} are from an experiment where we add latency between clusters to simulate network latency between clusters.
Figure~\ref{fig:throughput_drwt_skew} shows the average throughput of distributed read-write transactions. We notice that the throughput decreases as the data skew moves from read to write-intensive transactions. This is expected as the write-intensive transactions require more coordination between clusters. The range of the read and write operations in these experiments is selected to ensure that each transaction reads or writes some data on each participating cluster. In these figures, we notice that \textit{"R=5,W=1"}, essentially means local-read-write transactions. Thus, this experiment also provides insight into the cost of coordination in read-write transactions. The coordination cost is seen in figure~\ref{fig:latency_drwt_skew}. We notice the increase in latency as the operation skews towards write operations: meaning transactions access more clusters.

Figure~\ref{fig:throughput_drwt_latency} shows the variation in throughput as network latency between clusters increases. We change latency between clusters by $0, 20, 70, 150, 300, 500$ milliseconds. This allows us to simulate geo-distributed participating nodes. We see that the throughput drops considerably as the network latency increases. This is due to the coordination-intensive 2PC used by TransEdge in executing distributed read-write transactions. Figure~\ref{fig:transedge-aborts-drwt} shows the aborts that result as a consequence of the increase in network latency between clusters. Figure~\ref{fig:lrwt-drwt-skew}, shows the performance of mixed workload in where the read-write transactions are varied from a those affecting only a single cluster to a coordination-intensive workload. The graphs in figure\ref{fig:lrwt-drwt-skew} show a much higher throughput for a workload affecting a single cluster (LRWT=100\%, DRWT=0\%). This is because this workload requires no coordination with any other cluster. The lowest throughput is seen in workload comprising 100\% distributed read-write transactions(DRWT) as they pay a much higher cost of 2PC coordination. 

Figure~\ref{fig:drwt-multicluster}, shows the impact of changing the replica size to support multiple levels of Byzantine faults varying from $f=1$ to $f=3$, thus the number of replicas per cluster changes from $4$ to $10$ respectively. We notice that the smaller the number of replicas in a cluster the higher is the throughput. This is due to the reduced cost of intra-cluster coordination to execute read-write transactions.

\cut{In this experiment, there are 10--15\% transactions that are local and the rest are distributed read-write transactions. The throughput results show that the peak throughput is achieved at 700 transactions per batch. With smaller batch sizes, the overhead of coordination prevents committing enough batches to lead to utilizing resources. However, as the batch size increases, the overhead of coordination is amortized and throughput is increased, indicating that compute and communication resources are better utilized. 
}

%
%

\cut{
At lower values of transactions per batch the transaction processor pipeline (local to the cluster) processes these transactions and transactions that require two-phase commit are pushed into the second multi-threaded transactions pipeline. The design of the two-phase commit processor does not allow transactions in later batches to be processed before the those in earlier batches have been processed. This causes transactions to be queued in the two phase commit phase. The values of throughput at lower transaction per batch captures this queue. At higher values of batch size the throughput does not fall as low as at low batch sizes because at higher batch sizes a very large number of transactions are processed by the local transaction processor at remote clusters as well. Thus, the two-phase prepare queue does not form.
Figure~\ref{fig:lrwt-drwt-skew}, shows variations in throughput of a mixed workload where we vary the ratio of local-read-write transactions and distributed read-write transactions. The figure shows the drop in throughput as we increase the ratio of distributed read-write transactions. The drop in throughput is due to the coordination cost of 2PC used in the distributed read-write transactions.
}

\begin{figure}[!t]
\begin{center}
\includegraphics[scale=0.45]{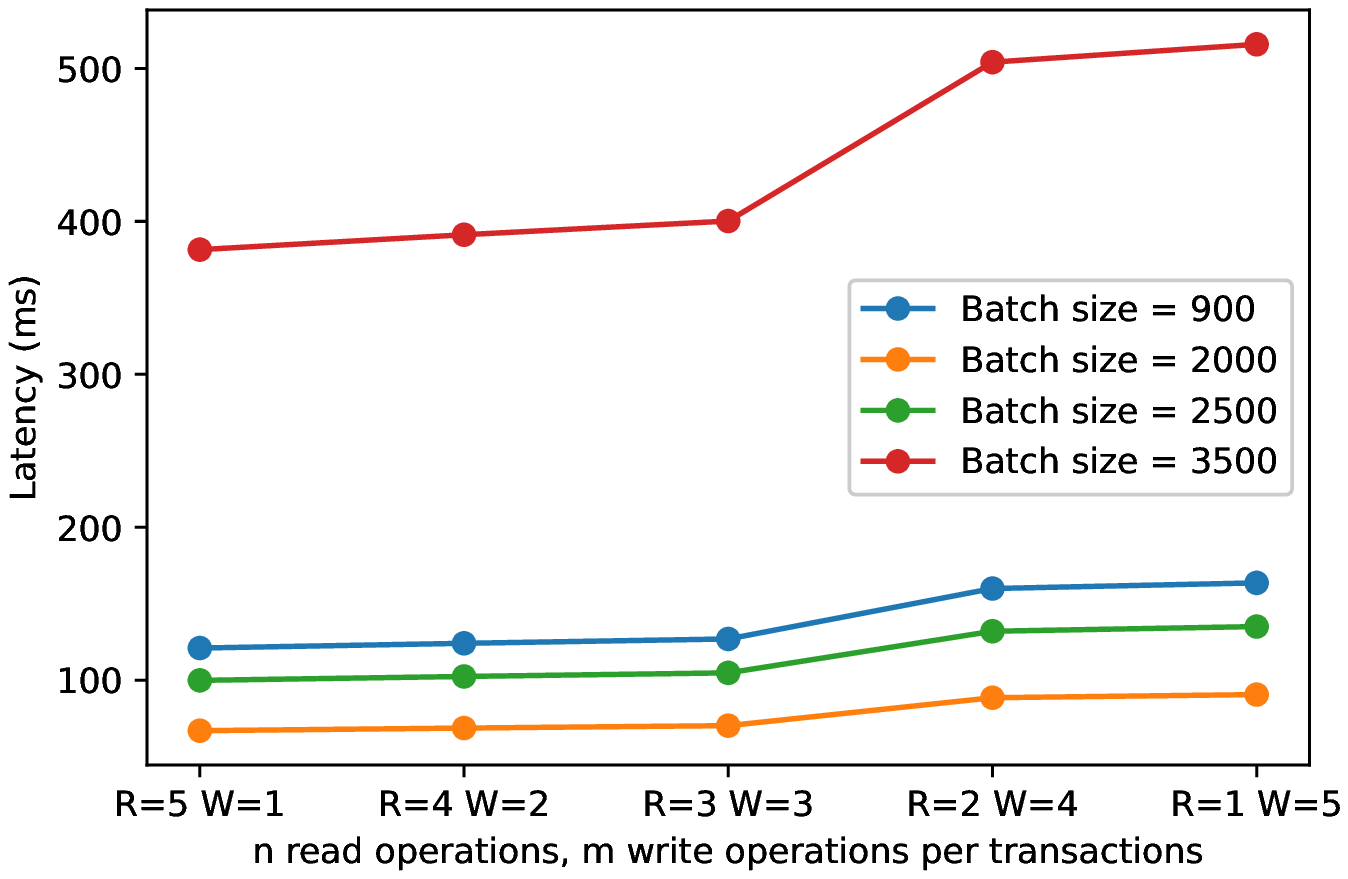}
\caption{Average latency of distributed read-write transactions. The horizontal axis shows the variation of data skew within the transactions.} 
\label{fig:latency_drwt_skew}
\end{center}
\end{figure}

\begin{figure}[!t]
\begin{center}
\includegraphics[scale=0.45]{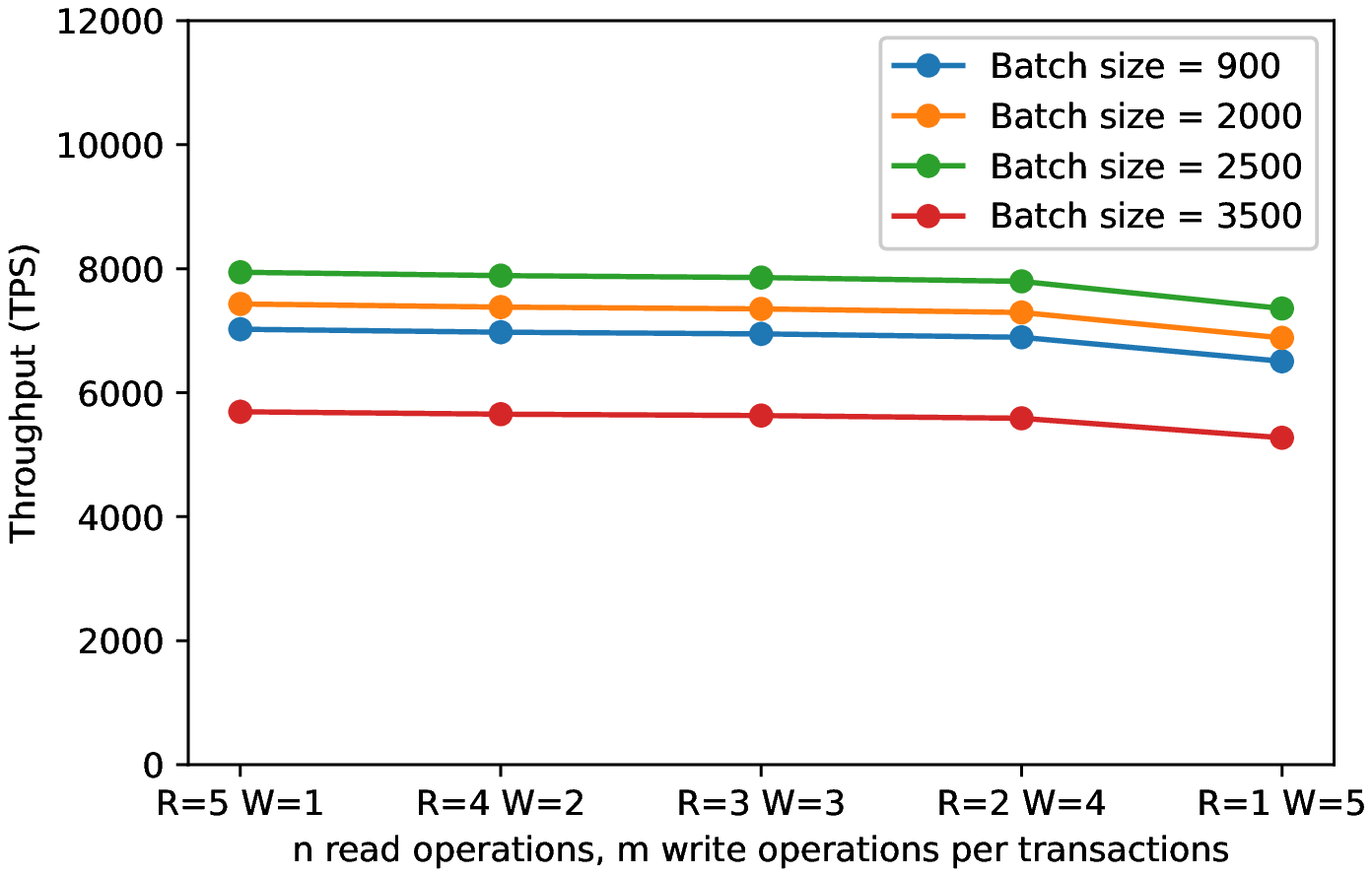}
\caption{Average Throughput of distributed read-write transactions. The horizontal axis shows the variation of data skew within the transactions.} 
\label{fig:throughput_drwt_skew}
\end{center}
\end{figure}

\begin{figure}[!t]
\begin{center}
\includegraphics[scale=0.45]{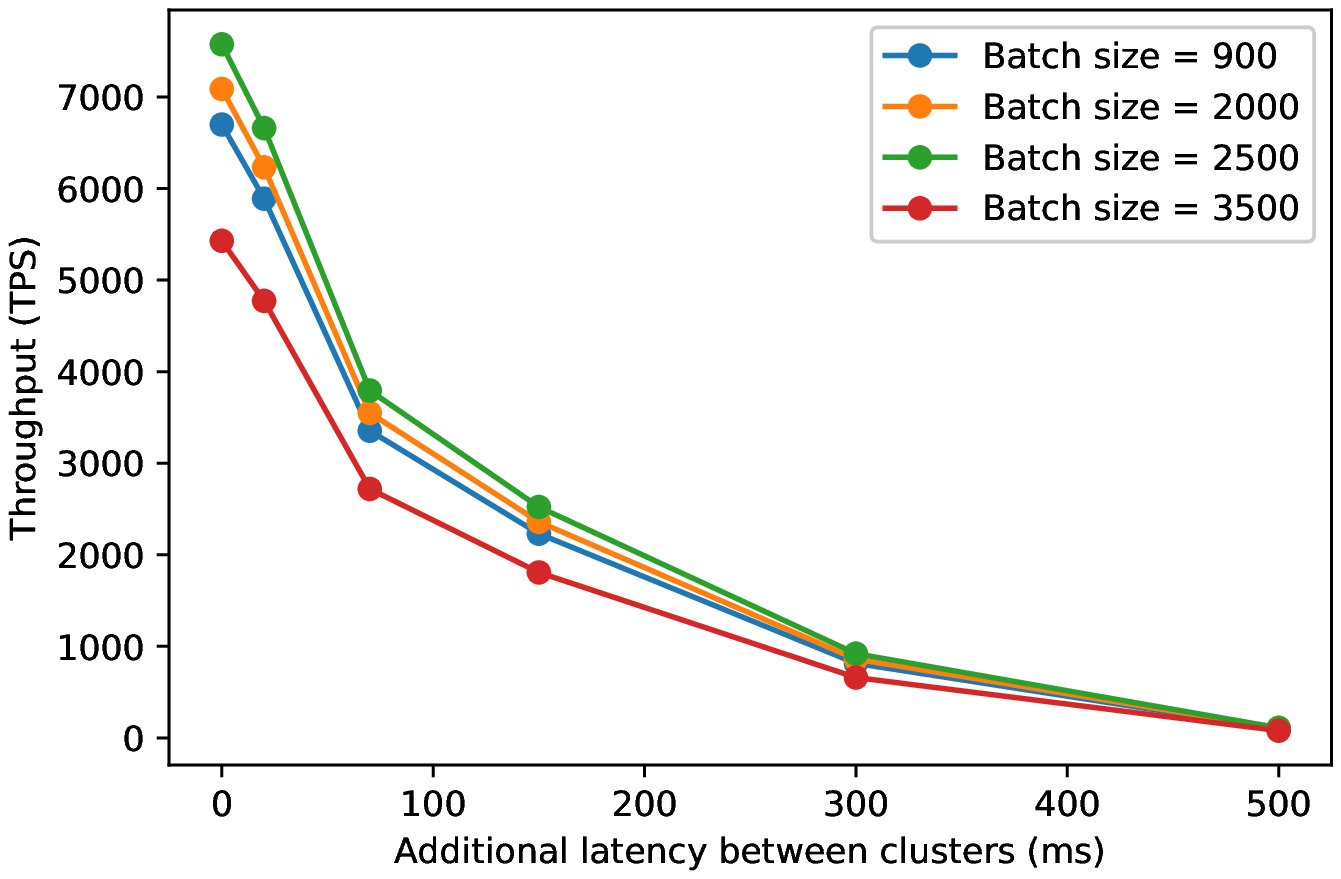}
\caption{Average Throughput of distributed read-write transactions when additional network latency is added between clusters varying between 0ms to 500ms} 
\label{fig:throughput_drwt_latency}
\end{center}
\end{figure}

\begin{figure}[h]
\begin{center}
\includegraphics[scale=0.45]{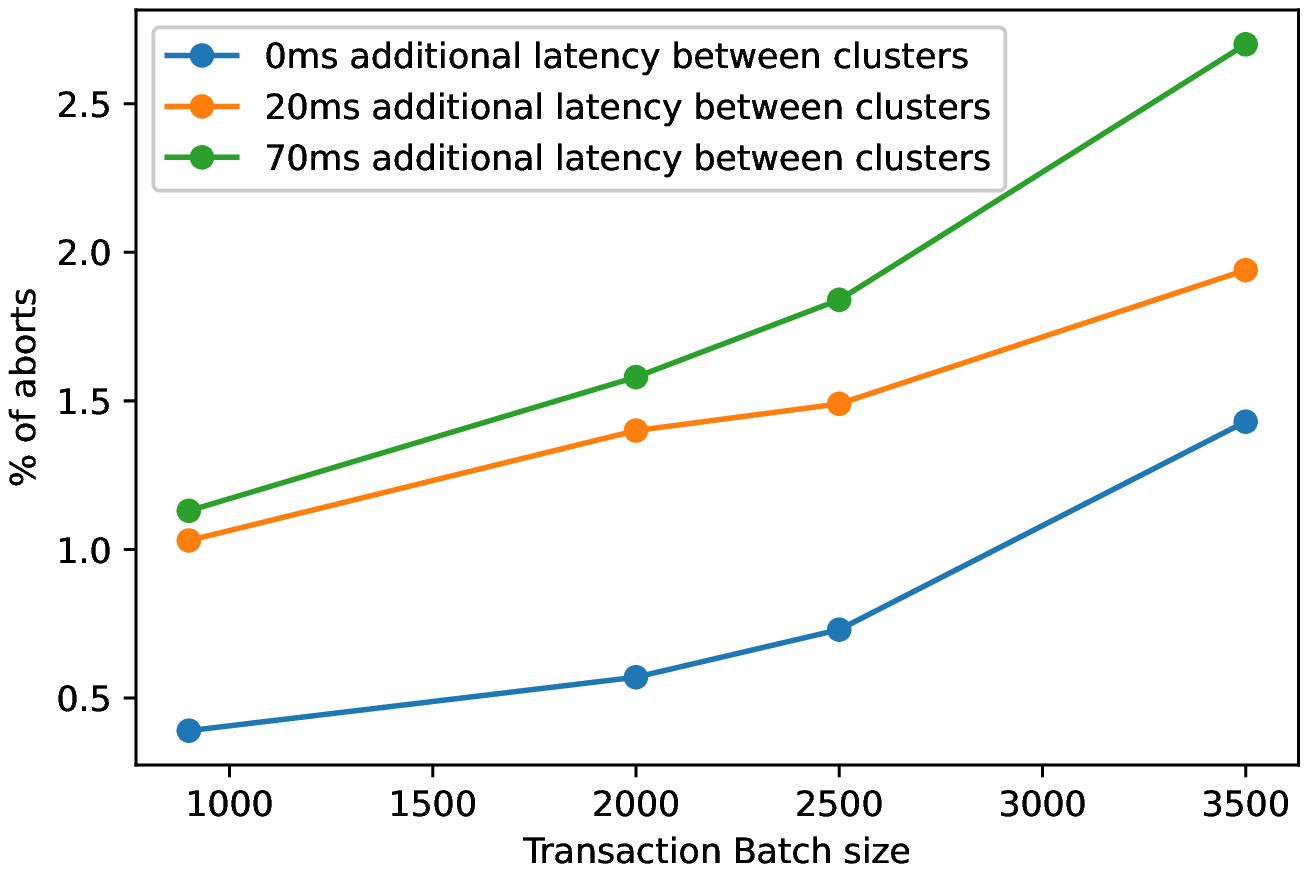}
\caption{Percentage of aborts in read-write transactions in TransEdge.}
\label{fig:transedge-aborts-drwt}
\end{center}
\end{figure}

\begin{figure}[h]
\begin{center}
\includegraphics[scale=0.45]{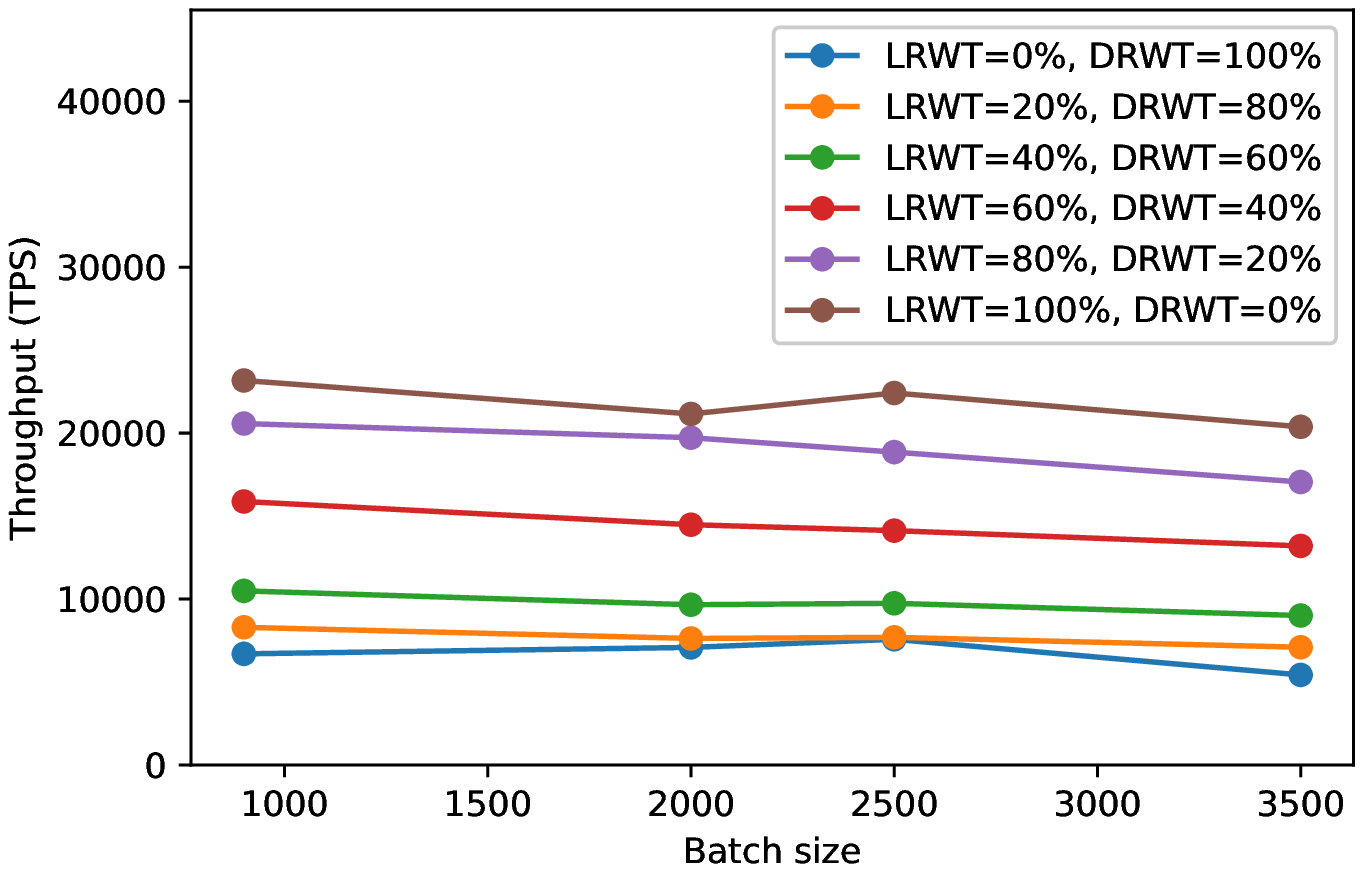}
\caption{Variation in throughput when the ratio of Local read-write and Distributed read-write transactions is changed in the workload.}
\label{fig:lrwt-drwt-skew}
\end{center}
\end{figure}

\begin{figure}[h]
\begin{center}
\includegraphics[scale=0.45]{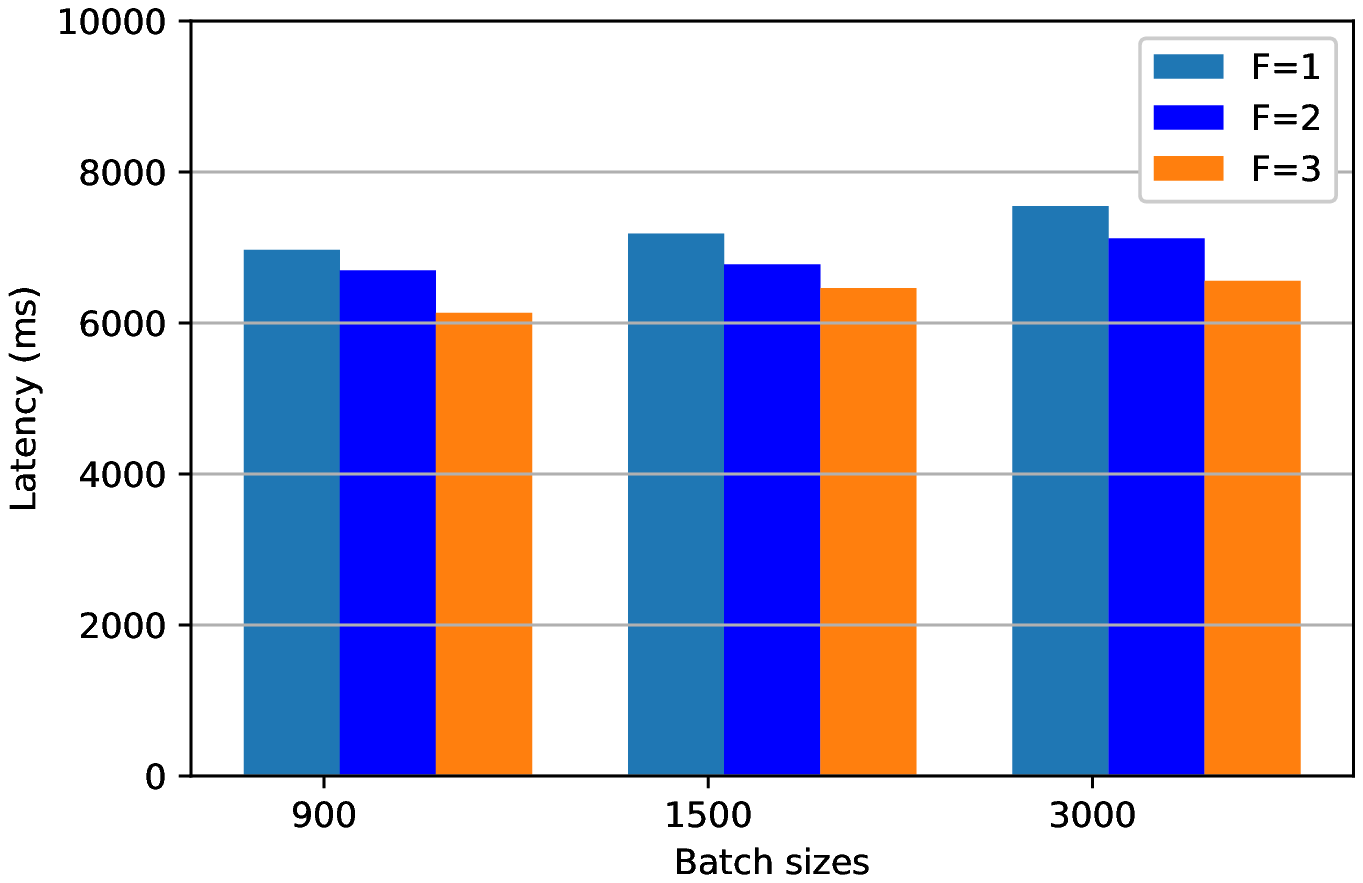}
\caption{Variation in throughput when the replicas per cluster are increased.}
\label{fig:drwt-multicluster}
\end{center}
\end{figure}


\cut{
At low values of batch sizes the latency of distributed transactions is seen to be very high. This is due to the number of batches queued for the two-phase commit phase at lower batch sizes. As the batch size increases the number of transactions processed per batch increases and this lowers the wait time for batches in the two-phase commit queue. As in the throughput graphs the latency at higher batch sizes does not degrade to the levels of those at lower batch sizes. The comparatively lower values of latency at higher batch sizes again shows the advantage of batching processing of distributed transactions.
}

\section{Related Work}\label{sec:related}
\cut{\vspace{0.1in}}

\cut{In this section, we overview related work in the areas of
byzantine agreement and trusted transaction processing.}

\subsection{Byzantine Agreement}
\label{sub:related-byz}
\cut{\vspace{0.1in}}

The byzantine agreement problem was proposed in the early
1980s~\cite{lamport1982byzantine,pease1980reaching}. A notable
milestone since then has been the proposal of the PBFT
protocol~\cite{castro1999practical} that we build upon in this paper.
In
the decade following the publication of PBFT, there has been a
resurgence of byzantine fault-tolerance
protocols~\cite{kotla2007zyzzyva,cowling2006hq,abd2005fault,kotla2004high}.
Byzantine agreement is getting renewed interest due to its
applications in blockchain
technology~\cite{nakamoto2008bitcoin,eyal2016bitcoin,gilad2017algorand,miller2016honey,kogias2016enhancing,pass2017hybrid,mazieres2015stellar}.
Byzantine agreement is especially relevant for \emph{permissioned
blockchain} where the set of writers to the blockchain are known but
potentially malicious. The interest in permissioned blockchain
technologies is due to various data management applications, such as
supply-chain management. This led to a number of
permissioned blockchain
systems~\cite{androulaki2018hyperledger,quorum,chain,parity,ripple}.
Because these are permissioned blockchain systems, they can use
traditional byzantine agreement protocols as their
agreement/consensus component. 

\cut{
BFT-SMaRt\cite{bft-smart-lib} is an open source Java library which provides byzantine fault tolerant consensus and state-machine replication services to applications. It has been used by several works\cite{al2017chainspace,androulaki2018hyperledger} for testing prototypes and has been used in TransEdge for prototype implementation. We rely on the version of PBFT implemented by BFT-SMaRt library for supporting byzantine fault tolerance in TransEdge clusters.
}

\cut{
GBFT's contribution compared to other BFT systems is that it is
designed for a GEDM environment with a large number of nodes that are
distributed around the world. GBFT proposes a locality-aware
hierarchical design and an asymmetric hybrid fault-tolerance model
that leverages the trusted cloud node. 
The interest in global-scale
systems~\cite{spanner,BBC11,noghabi2016ambry} motivated work on
consensus~\cite{L01,lamport2006fast,howard2016flexible} across
wide-area links~\cite{paxoscp,mdcc,moraru2013there,dpaxos,wpaxos}.
Unlike these aforementioned works, GBFT tolerates byzantine failures.
}

In global-scale environments, where nodes are separated by wide-area
latency, BFT systems incur significant overhead due to the many
rounds of communication needed to commit a request. To overcome this,
hierarchical BFT systems were
proposed~\cite{amir2010steward,nawab2019blockplane,amiri2019parblockchain,gupta13resilientdb}.
%
TransEdge is closest to this body of work. Its contribution to
hierarchical BFT systems is a design that is centered around
providing efficient read-only transaction processing that is
commit-free and non-interfering.

One of the design issues faced by hierarchical BFT systems is the grouping of nodes to form partitions. Mechanisms for nodes grouping depend on the use case and the considered system environment. These mechanisms are typically through manual administrator intervention, a placement/configuration protocol, or a distributed/decentralized membership mechanism. The goal of these grouping mechanisms is to ensure that no more than $f$ malicious nodes can exist in each group. For system environments that are permissioned~\cite{gupta13resilientdb,nawab2019blockplane,amiri2019parblockchain,hellings2021byshard}, this is ensured by making each cluster have no more than $f$ nodes that are independent. Independence here refers to the property that a failure (byzantine or otherwise) of one of the $f$ nodes is not going to lead to the failure of another node in the cluster. Ensuring this is application specific and can be performed during setup by authenticating each permissioned participant. In open membership system environments, various methods can be used to select a group of nodes that probabilistically guarantees that no more than $f$ malicious nodes are selected. This includes using reputation-based byzantine mechanisms~\cite{lei2018reputation,yuan2021efficient,zegers2020reputation} that can be utilized to select a grouping of nodes based on their past behavior (reputation).
Another method utilizes randomized methods such as Verified Random Functions (VRF)~\cite{gilad2017algorand,micali1999verifiable} where $m$ nodes are randomly selected in a decentralized way from a larger pool of $n$ nodes. These methods, however, need to be adapted to hierarchical latency-sensitive systems to balance the randomness of the grouping of nodes and the proximity of nodes in a cluster. Additionally, they need to be extended to enable selecting multiple clusters, one for each shard of the data.


\subsection{Read-only Transactions}

Read-only transactions have been a topic of interest for a long time~\cite{garciamolina1982, satyanarayanan1993efficient, lu2020performance, DBLP:journals/corr/abs-2107-11144,lloydSnow2016}. TransEdge builds on these works to construct a consistent read-only transaction algorithm suited for byzantine fault-tolerant systems. Recently, there have been some advances to try to formalize the properties of read-only transactions~\cite{lu2020performance,lloydSnow2016}. Most notable is the SNOW theorem~\cite{lloydSnow2016} that allows us to reason about the properties of read-only transactions in a distributed system. The SNOW impossibility result does not allow a system to simultaneously support all the SNOW properties. TransEdge supports non-blocking read-only transactions, read-write transactions are conflict-serializable and TransEdge allows read-write transactions to coexist with read-only transactions. However, TransEdge does not support one-round read-only transactions as TransEdge requires two rounds in the worst case to execute consistent read-only transactions. 
What distinguishes TransEdge is its focus on a byzantine environment, where nodes can act maliciously. 
TransEdge provides mechanisms to ensure the authenticity of responses that involve the use of authenticated data structures that verifies the integrity of responses. Doing this by itself is insufficient as the performance overhead can be high. TransEdge extends these trusted mechanisms with dependency tracking to enable fast and efficient processing while maintaining the integrity of responses.

Augustus~\cite{padilhaaugustus2013} is a system that deals with similar challenges to TransEdge: serializable transactions in a BFT environment and support of read-only transactions. However, Augustus uses shared locks for read-only transactions, causing read-only transactions to interfere with read-write transactions. TransEdge, on the other hand, does not use locks during read-only transactions and therefore ensures non-interference (in terms of conflicts) with read-write transactions. Augustus requires voting from participating replicas which adds to the overhead of read-only transactions. TransEdge requires the response from a single node per partition and does not need to involve other participants in the read-only transaction.


\subsection{Trusted Transaction Processing}
\label{sub:related-trusted}
TransEdge's trusted transaction processing gets inspiration from prior
work in using Authenticated Data Structures (ADSs)~\cite{merkle1980protocols} in
transaction processing. 
%
ADSs are data structures that are capable of providing a proof of the
authenticity of the stored data. 
ADSs have been used for databases
as a solution to the problem of outsourcing databases to public cloud
providers. Some of these solutions focus on query
processing~\cite{zhang2017vsql}. More related to TransEdge is usage of ADSs for
trusted transaction processing~\cite{jain2013trustworthy}. Unlike prior work in ADSs for
transaction processing, TransEdge tackles a different system
model where the untrusted nodes are many edge nodes around the world
instead of a node in the cloud. In terms of functionality, TransEdge
supports updating the ADS through a set of untrusted nodes using BFT
replication, whereas prior solutions rely on a trusted node that
recomputes the Merkle tree (an infeasible design for TransEdge since a
trusted node is not available.) 

BlockchainDB\cite{blockchaindb2019} is another work that deals with building a scalable database on top of a blockchain layer. BlockchainDB supports eventual and sequential consistency as the system architecture does not support serializable transactional workloads. BlockchainDB allows clients to verify if the operation is executed on the blockchain. This verification process requires querying the majority of the peers on the network. In TransEdge, the execution of read-write transactions ensures that signatures are shared across participating nodes and are part of the log. Thus, we enable single-round read verification however at the cost of storage resulting in a much faster verification process.

\cut{
Recently, the interest in permissioned blockchain led to an interest
in building transactional abstraction on top of BFT
replication~\cite{}. The closes work to ours is Chainspace~\cite{},
which---like TransEdge---layers 2PC over a byzantine replication
component. TransEdge---unlike Chainspace---supports efficient snapshot
read-only transactions, whereas Chainspace has no special support for
read-only transactions. Furthermore, TransEdge utilizes GBFT that
enables efficient and flexible operation of BFT clusters in GEDM,
whereas Chainspace utilizes an existing BFT replication layer with no
such support.
}

\section{Conclusion}
\label{sec:conclusion}

In this paper, we introduce TransEdge, a trusted distributed transaction processing protocols for Global-Edge Data Management (GEDM). TransEdge's main goal is to provide efficient support for snapshot read-only transactions. To this end, TransEdge builds on hierarchical BFT systems and extends them with dependency tracking mechanisms that are trusted. This involves redesigning hierarchical commit protocols and augmenting and managing meta-information such as dependency vectors and the use of Authenticated Data Structures (ADSs). Our evaluation shows that TransEdge can perform distributed read-only transactions efficiently and $9x-24x$ faster 
than running them as regular transactions.
\vspace{0.1in}
\section{Acknowledgments}

This research is supported in part by the NSF under grant CNS-1815212 and a gift from Facebook.

\bibliographystyle{abbrv}
\bibliography{citations}

\begin{thebibliography}{10}

\bibitem{chain}
{Chain}.
\newblock \url{http://chain.com/}.

\bibitem{parity}
{Ethcore. Parity: next generation ethereum browser}.
\newblock \url{https://ethcore.io/parity.html}.

\bibitem{quorum}
{Quorum}.
\newblock \url{http://www.jpmorgan.com/global/Quorum}.

\bibitem{ripple}
{Ripple}.
\newblock \url{https://ripple.com}.

\bibitem{abd2005fault}
M.~Abd-El-Malek, G.~R. Ganger, G.~R. Goodson, M.~K. Reiter, and J.~J. Wylie.
\newblock Fault-scalable byzantine fault-tolerant services.
\newblock {\em ACM SIGOPS Operating Systems Review}, 39(5):59--74, 2005.

\bibitem{agrawal1987distributed}
D.~Agrawal, A.~J. Bernstein, P.~Gupta, and S.~Sengupta.
\newblock Distributed optimistic concurrency control with reduced rollback.
\newblock {\em Distributed Computing}, 2(1):45--59, 1987.

\bibitem{al2017chainspace}
M.~Al-Bassam, A.~Sonnino, S.~Bano, D.~Hrycyszyn, and G.~Danezis.
\newblock Chainspace: A sharded smart contracts platform.
\newblock {\em arXiv preprint arXiv:1708.03778}, 2017.

\bibitem{amir2010steward}
Y.~Amir, C.~Danilov, D.~Dolev, J.~Kirsch, J.~Lane, C.~Nita-Rotaru, J.~Olsen,
  and D.~Zage.
\newblock Steward: Scaling byzantine fault-tolerant replication to wide area
  networks.
\newblock {\em IEEE Transactions on Dependable and Secure Computing},
  7(1):80--93, 2010.

\bibitem{amiri2019parblockchain}
M.~J. Amiri, D.~Agrawal, and A.~El~Abbadi.
\newblock Parblockchain: Leveraging transaction parallelism in permissioned
  blockchain systems.
\newblock In {\em 2019 IEEE 39th International Conference on Distributed
  Computing Systems (ICDCS)}, pages 1337--1347. IEEE, 2019.

\bibitem{androulaki2018hyperledger}
E.~Androulaki, A.~Barger, V.~Bortnikov, C.~Cachin, K.~Christidis, A.~De~Caro,
  D.~Enyeart, C.~Ferris, G.~Laventman, Y.~Manevich, et~al.
\newblock Hyperledger fabric: a distributed operating system for permissioned
  blockchains.
\newblock {\em arXiv preprint arXiv:1801.10228}, 2018.

\bibitem{DBLP:journals/corr/abs-2107-11144}
C.~Berger, H.~P. Reiser, and A.~Bessani.
\newblock Making reads in {BFT} state machine replication fast, linearizable,
  and live.
\newblock {\em CoRR}, abs/2107.11144, 2021.

\bibitem{bernstein1987concurrency}
P.~A. Bernstein, V.~Hadzilacos, and N.~Goodman.
\newblock {\em Concurrency control and recovery in database systems}, volume
  370.
\newblock Addison-wesley Reading, 1987.

\bibitem{bft-smart-lib}
A.~{Bessani}, J.~{Sousa}, and E.~E.~P. {Alchieri}.
\newblock State machine replication for the masses with bft-smart.
\newblock In {\em 2014 44th Annual IEEE/IFIP International Conference on
  Dependable Systems and Networks}, pages 355--362, 2014.

\bibitem{bronson2013tao}
N.~Bronson, Z.~Amsden, G.~Cabrera, P.~Chakka, P.~Dimov, H.~Ding, J.~Ferris,
  A.~Giardullo, S.~Kulkarni, H.~Li, et~al.
\newblock $\{$TAO$\}$: Facebook’s distributed data store for the social
  graph.
\newblock In {\em 2013 $\{$USENIX$\}$ Annual Technical Conference
  ($\{$USENIX$\}$$\{$ATC$\}$ 13)}, pages 49--60, 2013.

\bibitem{castro1999correctness}
M.~Castro, B.~Liskov, et~al.
\newblock A correctness proof for a practical byzantine-fault-tolerant
  replication algorithm.
\newblock Technical report, Technical Memo MIT/LCS/TM-590, MIT Laboratory for
  Computer Science, 1999.

\bibitem{castro1999practical}
M.~Castro, B.~Liskov, et~al.
\newblock Practical byzantine fault tolerance.
\newblock In {\em OSDI}, volume~99, pages 173--186, 1999.

\bibitem{ycsb2010}
B.~F. Cooper, A.~Silberstein, E.~Tam, R.~Ramakrishnan, and R.~Sears.
\newblock Benchmarking cloud serving systems with ycsb.
\newblock In {\em Proceedings of the 1st ACM Symposium on Cloud Computing},
  SoCC '10, page 143–154, New York, NY, USA, 2010. Association for Computing
  Machinery.

\bibitem{spanner}
J.~Corbett et~al.
\newblock Spanner: Google's globally-distributed database.
\newblock {\em OSDI}, 2012.

\bibitem{cowling2006hq}
J.~Cowling, D.~Myers, B.~Liskov, R.~Rodrigues, and L.~Shrira.
\newblock Hq replication: A hybrid quorum protocol for byzantine fault
  tolerance.
\newblock In {\em Proceedings of the 7th symposium on Operating systems design
  and implementation}, pages 177--190. USENIX Association, 2006.

\bibitem{das2011albatross}
S.~Das, S.~Nishimura, D.~Agrawal, and A.~El~Abbadi.
\newblock Albatross: Lightweight elasticity in shared storage databases for the
  cloud using live data migration.
\newblock {\em Proceedings of the VLDB Endowment}, 4(8):494--505, 2011.

\bibitem{blockchaindb2019}
M.~El-Hindi, C.~Binnig, A.~Arasu, D.~Kossmann, and R.~Ramamurthy.
\newblock Blockchaindb: A shared database on blockchains.
\newblock {\em Proc. VLDB Endow.}, 12(11):1597–1609, jul 2019.

\bibitem{eyal2016bitcoin}
I.~Eyal, A.~E. Gencer, E.~G. Sirer, and R.~Van~Renesse.
\newblock Bitcoin-ng: A scalable blockchain protocol.
\newblock In {\em NSDI}, pages 45--59, 2016.

\bibitem{garciamolina1982}
H.~Garcia-Molina and G.~Wiederhold.
\newblock Read-only transactions in a distributed database.
\newblock {\em ACM Trans. Database Syst.}, 7(2):209–234, June 1982.

\bibitem{gilad2017algorand}
Y.~Gilad, R.~Hemo, S.~Micali, G.~Vlachos, and N.~Zeldovich.
\newblock Algorand: Scaling byzantine agreements for cryptocurrencies.
\newblock In {\em Proceedings of the 26th Symposium on Operating Systems
  Principles}, pages 51--68. ACM, 2017.

\bibitem{gupta13resilientdb}
S.~Gupta, S.~Rahnama, J.~Hellings, and M.~Sadoghi.
\newblock Resilientdb: Global scale resilient blockchain fabric.
\newblock {\em Proceedings of the VLDB Endowment}, 13(6), 2020.

\bibitem{hellings2021byshard}
J.~Hellings and M.~Sadoghi.
\newblock Byshard: Sharding in a byzantine environment.
\newblock {\em Proceedings of the VLDB Endowment}, 14(11):2230--2243, 2021.

\bibitem{jain2013trustworthy}
R.~Jain and S.~Prabhakar.
\newblock Trustworthy data from untrusted databases.
\newblock In {\em 2013 IEEE 29th International Conference on Data Engineering
  (ICDE)}, pages 529--540. IEEE, 2013.

\bibitem{keahey2020lessons}
K.~Keahey, J.~Anderson, Z.~Zhen, P.~Riteau, P.~Ruth, D.~Stanzione, M.~Cevik,
  J.~Colleran, H.~S. Gunawi, C.~Hammock, J.~Mambretti, A.~Barnes, F.~Halbach,
  A.~Rocha, and J.~Stubbs.
\newblock Lessons learned from the chameleon testbed.
\newblock In {\em Proceedings of the 2020 USENIX Annual Technical Conference
  (USENIX ATC '20)}. USENIX Association, July 2020.

\bibitem{kogias2016enhancing}
E.~K. Kogias, P.~Jovanovic, N.~Gailly, I.~Khoffi, L.~Gasser, and B.~Ford.
\newblock Enhancing bitcoin security and performance with strong consistency
  via collective signing.
\newblock In {\em 25th USENIX Security Symposium (USENIX Security 16)}, pages
  279--296, 2016.

\bibitem{kotla2007zyzzyva}
R.~Kotla, L.~Alvisi, M.~Dahlin, A.~Clement, and E.~Wong.
\newblock Zyzzyva: speculative byzantine fault tolerance.
\newblock In {\em ACM SIGOPS Operating Systems Review}, volume~41, pages
  45--58. ACM, 2007.

\bibitem{kotla2004high}
R.~Kotla and M.~Dahlin.
\newblock High throughput byzantine fault tolerance.
\newblock In {\em Dependable Systems and Networks, 2004 International
  Conference on}, pages 575--584. IEEE, 2004.

\bibitem{kung1981optimistic}
H.-T. Kung and J.~T. Robinson.
\newblock On optimistic methods for concurrency control.
\newblock {\em ACM Transactions on Database Systems (TODS)}, 6(2):213--226,
  1981.

\bibitem{lamport1982byzantine}
L.~Lamport, R.~Shostak, and M.~Pease.
\newblock The byzantine generals problem.
\newblock {\em ACM Transactions on Programming Languages and Systems (TOPLAS)},
  4(3):382--401, 1982.

\bibitem{lei2018reputation}
K.~Lei, Q.~Zhang, L.~Xu, and Z.~Qi.
\newblock Reputation-based byzantine fault-tolerance for consortium blockchain.
\newblock In {\em 2018 IEEE 24th international conference on parallel and
  distributed systems (ICPADS)}, pages 604--611. IEEE, 2018.

\bibitem{lloydSnow2016}
H.~Lu, C.~Hodsdon, K.~Ngo, S.~Mu, and W.~Lloyd.
\newblock The {SNOW} theorem and latency-optimal read-only transactions.
\newblock In {\em 12th {USENIX} Symposium on Operating Systems Design and
  Implementation ({OSDI} 16)}, pages 135--150, Savannah, GA, Nov. 2016.
  {USENIX} Association.

\bibitem{lu2020performance}
H.~Lu, S.~Sen, and W.~Lloyd.
\newblock Performance-optimal read-only transactions.
\newblock In {\em 14th $\{$USENIX$\}$ Symposium on Operating Systems Design and
  Implementation ($\{$OSDI$\}$ 20)}, pages 333--349, 2020.

\bibitem{mazieres2015stellar}
D.~Mazieres.
\newblock The stellar consensus protocol: A federated model for internet-level
  consensus.
\newblock {\em Stellar Development Foundation}, 2015.

\bibitem{merkle1980protocols}
R.~C. Merkle.
\newblock Protocols for public key cryptosystems.
\newblock In {\em 1980 IEEE Symposium on Security and Privacy}, pages 122--122.
  IEEE, 1980.

\bibitem{micali1999verifiable}
S.~Micali, M.~Rabin, and S.~Vadhan.
\newblock Verifiable random functions.
\newblock In {\em Foundations of Computer Science, 1999. 40th Annual Symposium
  on}, pages 120--130. IEEE, 1999.

\bibitem{miller2016honey}
A.~Miller, Y.~Xia, K.~Croman, E.~Shi, and D.~Song.
\newblock The honey badger of bft protocols.
\newblock In {\em Proceedings of the 2016 ACM SIGSAC Conference on Computer and
  Communications Security}, pages 31--42. ACM, 2016.

\bibitem{nakamoto2008bitcoin}
S.~Nakamoto.
\newblock Bitcoin: A peer-to-peer electronic cash system.
\newblock 2008.

\bibitem{nawab2019blockplane}
F.~Nawab and M.~Sadoghi.
\newblock Blockplane: A global-scale byzantizing middleware.
\newblock In {\em 2019 IEEE 35th International Conference on Data Engineering
  (ICDE)}, pages 124--135. IEEE, 2019.

\bibitem{padilhaaugustus2013}
R.~Padilha and F.~Pedone.
\newblock Augustus: Scalable and robust storage for cloud applications.
\newblock In {\em Proceedings of the 8th ACM European Conference on Computer
  Systems}, EuroSys '13, page 99–112, New York, NY, USA, 2013. Association
  for Computing Machinery.

\bibitem{pass2017hybrid}
R.~Pass and E.~Shi.
\newblock Hybrid consensus: Efficient consensus in the permissionless model.
\newblock In {\em LIPIcs-Leibniz International Proceedings in Informatics},
  volume~91. Schloss Dagstuhl-Leibniz-Zentrum fuer Informatik, 2017.

\bibitem{pease1980reaching}
M.~Pease, R.~Shostak, and L.~Lamport.
\newblock Reaching agreement in the presence of faults.
\newblock {\em Journal of the ACM (JACM)}, 27(2):228--234, 1980.

\bibitem{satyanarayanan1993efficient}
O.~Satyanarayanan and D.~Agrawal.
\newblock Efficient execution of read-only transactions in replicated
  multiversion databases.
\newblock {\em Knowledge and Data Engineering, IEEE Transactions on},
  5(5):859--871, 1993.

\bibitem{schlageter1981optimistic}
G.~Schlageter.
\newblock Optimistic methods for concurrency control in distributed database
  systems.
\newblock In {\em Proceedings of the seventh international conference on Very
  Large Data Bases-Volume 7}, pages 125--130. VLDB Endowment, 1981.

\bibitem{weikum2001transactional}
G.~Weikum and G.~Vossen.
\newblock {\em Transactional information systems: theory, algorithms, and the
  practice of concurrency control and recovery}.
\newblock Elsevier, 2001.

\bibitem{yuan2021efficient}
X.~Yuan, F.~Luo, M.~Z. Haider, Z.~Chen, and Y.~Li.
\newblock Efficient byzantine consensus mechanism based on reputation in iot
  blockchain.
\newblock {\em Wireless Communications and Mobile Computing}, 2021, 2021.

\bibitem{zegers2020reputation}
F.~M. Zegers, M.~T. Hale, J.~M. Shea, and W.~E. Dixon.
\newblock Reputation-based event-triggered formation control and leader
  tracking with resilience to byzantine adversaries.
\newblock In {\em 2020 American Control Conference (ACC)}, pages 761--766.
  IEEE, 2020.

\bibitem{zhang2017vsql}
Y.~Zhang, D.~Genkin, J.~Katz, D.~Papadopoulos, and C.~Papamanthou.
\newblock vsql: Verifying arbitrary sql queries over dynamic outsourced
  databases.
\newblock In {\em 2017 IEEE Symposium on Security and Privacy (SP)}, pages
  863--880. IEEE, 2017.

\end{thebibliography}

\end{document}